\documentclass[11pt,a4paper]{article}
\usepackage[T1]{fontenc}
\usepackage{authblk}
\usepackage{graphicx}

\usepackage{amssymb}
\usepackage{hyperref}

\usepackage{amsfonts,amsmath,latexsym}
\usepackage[active]{srcltx}
\usepackage{enumitem}
\usepackage{color}
\newcommand{\bu}{\mbox{\boldmath{$u$}}}
\newcommand{\bt}{\mbox{\boldmath{$t$}}}
\newcommand{\bcero}{\mbox{\boldmath{$0$}}}

\newcommand{\bx}{\mbox{\boldmath{$\bx$}}}
\newcommand{\by}{\mbox{\boldmath{$\by$}}}
\newcommand{\bv}{\mbox{\boldmath{$v$}}}

\newcommand{\var}{\varepsilon}

\newcommand{\ba}{\mbox{\boldmath{$a$}}}

\newcommand{\fb}{{{f}}}
\newcommand{\bg}{\mbox{\boldmath{$g$}}}
\renewcommand{\by}{\mbox{\boldmath{$y$}}}
\renewcommand{\bx}{\mbox{\boldmath{$x$}}}
\newcommand{\be}{\mbox{\boldmath{$e$}}}
\newcommand{\bn}{\mbox{\boldmath{$n$}}}
\newcommand{\bh}{{{h}}}
\newcommand{\bbf}{\mbox{\boldmath{$f$}}}
\newcommand{\bbh}{\mbox{\boldmath{$h$}}}

\newcommand{\btheta}{\mbox{\boldmath{$\theta$}}}

\newcommand{\beeta}{\mbox{\boldmath{$\eta$}}}
\newcommand{\bxi}{\mbox{\boldmath{$\xi$}}}

\newcommand{\bzero}{\mbox{\boldmath{$0$}}}
\newcommand{\bTheta}{\mbox{\boldmath{$\Theta$}}}
\newcommand{\bUcal}{\mbox{\boldmath{$\mathcal{U}$}}}

\newcommand{\Gamae}{\Gamma^{\varepsilon }}

\renewcommand{\d}{\partial}
\newcommand{\eij}{e_{i||j}}

\newcommand{\ekl}{e_{k||l}}

\newcommand{\eab}{e_{\alpha||\beta}}
\newcommand{\est}{e_{\sigma||\tau}}

\newcommand{\estres}{e_{\sigma||3}}

\newcommand{\eatres}{e_{\alpha||3}}
\newcommand{\edtres}{e_{3||3}}

\newcommand{\gab}{\gamma_{\alpha\beta}}
\newcommand{\gst}{\gamma_{\sigma\tau}}
\newcommand{\rab}{\rho_{\alpha\beta}}
\newcommand{\rst}{\rho_{\sigma\tau}}
\renewcommand{\a}{a^{\alpha\beta\sigma\tau}}

\newcommand{\aeps}{a^{\alpha\beta\sigma\tau,\varepsilon}}

\newcommand{\en}{ \ \textrm{in} \ }
\newcommand{\on}{ \ \textrm{on} \ }
\newcommand{\into}{\int_{\omega}}
\newcommand{\intO}{\int_{\Omega}}
\newcommand{\intG}{\int_{\Gamma_+}}
\newcommand{\ten}{(a^{\alpha \sigma}a^{\beta \tau} + a^{\alpha \tau}a^{\beta\sigma})}

\newcommand{\Cof}{\mbox{\rm Cof\,}}
\newcommand{\meas}{\mbox{\rm meas\,}}

\usepackage{amsthm}

\newtheorem{theorem}{Theorem}[section]
\newtheorem{lemma}[theorem]{Lemma}
\newtheorem{problem}[theorem]{Problem}
\newtheorem{remark}[theorem]{Remark}

\begin{document}

\title{Models of Elastic Shells in Contact with a Rigid Foundation: An Asymptotic Approach}

%

\author{\'A. Rodr\'{\i}guez-Ar\'os\thanks{angel.aros@udc.es}}
\affil{Departamento de M\'etodos Matem\'aticos de Representaci\'on, Universidade da Coru\~na, Paseo de Ronda 51, 15011 A Coru\~na, Spain}




\maketitle

\begin{abstract}
We consider a family of linearly elastic shells with thickness $2\var$ (where $\var$ is a small parameter). The shells are clamped along a portion of their lateral face, all having the same middle surface $S$, and may enter in contact with a rigid foundation along the bottom face.

We are interested in studying the limit behavior of both the three-dimen\-sio\-nal problems, given in curvilinear coordinates, and their solutions (displacements $\bu^\var$ of covariant components $u_i^\var$) when $\var$ tends to  zero. To do that, we use asymptotic analysis methods.
On one hand, we find that if the applied body force density is $O(1)$ with respect to $\var$ and surface tractions density is $O(\var)$, a suitable approximation of the variational formulation of the contact problem is a two-dimensional variational inequality which can be identified as the variational formulation of the obstacle problem for an elastic membrane. On the other hand, if the applied body force density is $O(\var^2)$ and surface tractions density is $O(\var^3)$, the corresponding approximation is a different two-dimensional inequality which can be identified as the variational formulation of the obstacle problem for an elastic flexural shell. {We finally discuss the existence and uniqueness of solution for the limit two-dimensional variational problems found}.

\end{abstract}



\section{Introduction}
\label{intro}

In solid mechanics, the obtention of  models for rods, beams,
plates and shells is based on {\it a priori} hypotheses on the
displacement and/or stress fields which, upon substitution in the three-dimensional
equilibrium and constitutive equations, lead to useful simplifications. Nevertheless, from both
constitutive and geometrical point of views, there is a need to
justify the validity of most of the models obtained in this way.

For this reason a considerable effort has been made in the past decades by many authors in order
to derive new models and justify the existing ones by
using the asymptotic expansion method, whose foundations can be
found in \cite{Lions}. Indeed, the first applied results were obtained with the justification of the linearized theory of plate bending in \cite{CD,Destuynder}.

The theories of beam bending and rod stretching also benefited from the extensive use of asymptotic methods and so the justification of the Bernoulli-Navier model for the bending-stretching of elastic thin rods was provided in \cite{BV}. { Later}, the nonlinear case was studied
in \cite{CGLRT2} and the analysis and error estimation of higher-order terms in the asymptotic
expansion of the scaled unknowns was given in \cite{iv}. In \cite{TraViano}, the authors use the asymptotic method to justify the Saint-Venant, Timoshenko and Vlassov models of elastic beams.

A description of the  mathematical models for the three-dimensional elasticity, including the nonlinear aspects, together with a mathematical analysis of these models, can be found in \cite{Ciarlet2}. A justification of the two-dimensional equations of a linear plate can be found in  \cite{CD}.
An extensive review concerning  plate models can be found in \cite{Ciarlet3}, which also contains the justification of  the models by using asymptotic methods. The existence and uniqueness of solution of elliptic membrane shell equations, can be found in \cite{CiarletLods} and in \cite{CiarletLods2}.  These two-dimensional models are completely justified with convergence theorems. A complete theory regarding  elastic shells can be found in \cite{Ciarlet4b}, where models  for elliptic membranes, generalized membranes and flexural shells are presented. It contains a full description of the asymptotic procedure that leads to the corresponding sets of two-dimensional equations.

{ In the last decade, asymptotic methods have also been used to derive and justify contact models for beams (see \cite{Viano1,VAS_wear,ASV15}), plates (see \cite{Miara16}) and shallow shells \cite{LM1,LM2,LM}. Contact phenomena involving deformable bodies abound in industry and everyday life.} The contact of braking pads with wheels, a tire with a road, a piston with a skirt, a shoe with a floor, are just a few simple examples. For this reason, consi\-derable effort has been made with the modelling, analysis and numerical approximation of contact problems, and the engineering literature concerning this topic is rather extensive. An early attempt to the study of frictional contact problems within the framework of variational inequalities was made in \cite{DL}. Comprehensive references on analysis and numerical approximation of variational inequalities arising from contact problems include \cite{HS,HHNL,KO}. Also, from the numerical point of view, many algorithms have been developed in order to deal with the nonlinearities due to the contact conditions. Three main approaches are usually found in the literature; duality methods, combined with fixed point techniques (see \cite{BM}, \cite{GLT}), penalty methods (see \cite{KO}) and generalized Newton methods (see \cite{Kunisch}).  Mathematical, mechanical and numerical state of the art on the Contact Mechanics can be found in the proceedings \cite{MM,RJM}, in the special issue \cite{Sh} and in the monograph \cite{SST}, as well.

{ In the present paper we contribute to continue the pioneering work developed in \cite{LM1,LM2,LM} by removing the definitional restrictions of the shallow shells model. Moreover, in these works, the obstacle is a plane, while here we do not impose any particular shape. As a consequence, in this paper we find different limit models. Moreover, these limit models naturally describe the behaviour of elliptic membranes and flexural shells when in unilateral contact with a rigid foundation. To be precise the limit models can be identified as ``obstacle'' problems, since now the restrictions are to be imposed in the domain.}

{ Specifically,} we analyze the asymptotic behavior of the scaled three-dimen\-sional displacement field of an elastic shell in contact with a rigid foundation as the thickness of the shell approaches zero. We consider that the displacements vanish in a portion of the lateral face of the shell, obtaining the variational inequalities of an elastic membrane shell in unilateral contact, or of an elastic flexural shell in unilateral contact, depending on the order of the forces and the geometry. We will follow the notation and style of \cite{Ciarlet4b}, where the linear elastic shells without any contact conditions are studied. For this reason, we shall  reference auxiliary results which apply in the same manner to the unilateral contact case. One of the major differences with respect to previous works in elasticity without contact, consists in that the various problems are formulated by using inequalities instead of equalities. That will lead to a series of non trivial difficulties that need to be solved in order to find the zeroth-order approach of the solution.

Since the notation is well known in the community, let us advance the models obtained, for the benefit of the reader, { the details to be provided or referenced conveniently in the following sections}. Let $\omega$ be a domain in $\mathbb{R}^2$, the image of which by a sufficiently smooth function into $\mathbb{R}^3$ represents the middle surface $S$ of a shell of thickness $2\var$ in its natural state. This shell is made of an elastic homogeneous and isotropic material, it is fixed on a part of its lateral boundary and it is under the influence of volume forces and surface forces on its upper face. Additionally, the shell is in unilateral frictionless contact with a rigid foundation on the bottom face. Under these circumstances, we find the following models for describing the mechanical behaviour of the shell:

\paragraph{Elastic elliptic membrane in unilateral contact: Variational formulation.} Assume that $S$ is elliptic and the shell is fixed on the whole lateral face. Then:

Find $\bxi^\var=(\xi^\var_i):\omega\subset\mathbb{R}^2 \longrightarrow \mathbb{R}^3$ such that,
\begin{align*}
&\bxi^\var\in K_M(\omega):=\{\beeta=(\eta_i)\in H_0^1(\omega)\times H_0^1(\omega)\times L^2(\omega);\, \eta_3\ge0 \ \textrm{in} \ \omega \},\\
&\quad\var\int_{\omega} \aeps\gst(\bxi^\var)\gab(\beeta-\bxi^\var)\sqrt{a}dy\\
&\qquad   \ge\int_{\omega}p^{i,\var}(\eta_i-\xi_i^\var)\sqrt{a}dy \ \forall \beeta=(\eta_i)\in K_M(\omega),
\end{align*}
where the two-dimensional fourth-order elasticity tensor and the linearized change of metric tensor and applied forces are given, respectively as follows:
\begin{align*}
&a^{\alpha\beta\sigma\tau,\var}:=\frac{4\lambda\mu}{\lambda + 2\mu}a^{\alpha\beta}a^{\sigma\tau} + 2\mu \ten,\\
&\gab(\beeta):= \frac{1}{2}(\d_\alpha\eta_\beta + \d_\beta\eta_\alpha) - \Gamma_{\alpha\beta}^\sigma\eta_\sigma -b_{\alpha\beta}\eta_3,   \\
& p^{i,\var}:=\int_{-\var}^{\var}\fb^{i,\var}dx_3^\var +h_+^{i,\var}, \quad h_{+}^{i,\var}=\bh^{i,\var}(\cdot,\var).
\end{align*}

\paragraph{Elastic flexural shell in unilateral contact: Variational formulation.} Assume that the set $K_F(\omega)$ introduced below contains non trivial functions. Then:

Find $\bxi^\var=(\xi^\var_i):\omega\subset\mathbb{R}^2  \longrightarrow \mathbb{R}^3$ such that,
\begin{align*}
&\bxi^\var\in K_F(\omega):=\{ \beeta=(\eta_i) \in H^1(\omega)\times H^1(\omega)\times H^2(\omega) ;\\
&\qquad\qquad\qquad\qquad \eta_i=\d_\nu \eta_3=0 \en \gamma_0\subset\partial\omega,\ \gab(\beeta)=0,\ \eta_3\ge0 \en \omega\},\\
&\quad\frac{\var^3}{3}\int_{\omega} \aeps\rst(\bxi^\var)\rab(\beeta-\bxi^\var)\sqrt{a}dy\\
&\qquad\qquad\qquad\ge\int_{\omega}p^{i,\var}(\eta_i-\xi_i^\var)\sqrt{a}dy \ \forall \beeta=(\eta_i)\in K_F(\omega),
\end{align*}
where the two-dimensional linearized change of curvature tensor is given by:
\begin{align*}
&\rho_{\alpha\beta}(\beeta):= \d_{\alpha\beta}\eta_3 - \Gamma_{\alpha\beta}^\sigma \d_\sigma\eta_3 - b_\alpha^\sigma b_{\sigma\beta} \eta_3 + b_\alpha^\sigma (\d_\beta\eta_\sigma- \Gamma_{\beta\sigma}^\tau \eta_\tau)\\
&\qquad\qquad+ b_\beta^\tau(\d_\alpha\eta_\tau-\Gamma_{\alpha\tau}^\sigma\eta_\sigma ) + b^\tau_{\beta|\alpha} \eta_\tau.
\end{align*}

The structure of the paper is the following: in Section \ref{problema} we shall describe the variational and mechanical formulations of the problem in cartesian coordinates in the original domain and we reformulate the variational formulation in curvilinear coordinates. In Section \ref{seccion_dominio_ind} we will use a projection map into a reference domain { independent of $\var$} and we will introduce the scaled unknowns and forces as well as the assumptions on coefficients. We also devote this section to recall and derive results which will be needed in what follows. In Section \ref{procedure} we show the asymptotic analysis leading to the formulation of the variational inequalities of the elastic shells in unilateral contact. In Section \ref{Existencia} we show the re-scaled versions, with true physical meaning, of the variational formulations of the problems of elastic shells in unilateral contact, classified attending to their boundary conditions and the geometry of the middle surface $S$, { and discuss the existence and uniqueness of solution}.  Finally, in Section \ref{conclusiones} we shall present some conclusions and { describe the future work, namely the obtention of convergence results, to be provided in following papers}.

\section{The three-dimensional elastic shell contact problem} \label{problema}

We denote $\mathbb{S}^d$, where $d=2,3$ in practice, the space of second-order symmetric tensors on $\mathbb{R}^d$, while \textquotedblleft$\ \cdot$ \textquotedblright will represent the inner product and $|\cdot|$  the usual norm in $\mathbb{S}^d$ and  $\mathbb{R}^d$. In  what follows, unless the contrary is explicitly written, we will use summation convention on repeated indices. Moreover, Latin indices $i,j,k,l,...$, take their values in the set $\{1,2,3\}$, whereas Greek indices $\alpha,\beta,\sigma,\tau,...$, do it in the set  $\{1,2\}$. Also, we use standard notation for the Lebesgue and Sobolev spaces. Let $\omega$ be a domain of $\mathbb{R}^2$, with a Lipschitz-continuous boundary $\gamma=\partial\omega$. Let $\by=(y_\alpha)$ be a generic point of  its closure $\bar{\omega}$ and let $\d_\alpha$ denote the partial derivative with respect to $y_\alpha$.

Let $\btheta\in\mathcal{C}^2(\bar{\omega};\mathbb{R}^3)$ be an injective mapping such that the two vectors $\ba_\alpha(\by):= \d_\alpha \btheta(\by)$ are linearly independent. These vectors form the covariant basis of the tangent plane to the surface $S:=\btheta(\bar{\omega})$ at the point $\btheta(\by).$ We can consider the two vectors $\ba^\alpha(\by)$ of the same tangent plane defined by the relations $\ba^\alpha(\by)\cdot \ba_\beta(\by)=\delta_\beta^\alpha$, that constitute its contravariant basis. We define
\begin{align}\label{a_3}
\ba_3(\by)=\ba^3(\by):=\frac{\ba_1(\by)\wedge \ba_2(\by)}{| \ba_1(\by)\wedge \ba_2(\by)|},
\end{align}
the unit normal vector to $S$ at the point $\btheta(\by)$, where $\wedge$ denotes vector product in $\mathbb{R}^3.$

We can define the first fundamental form, given as metric tensor, in covariant or contravariant components, respectively, by
\begin{equation*}
a_{\alpha\beta}:=\ba_\alpha\cdot \ba_\beta, \qquad a^{\alpha\beta}:=\ba^\alpha\cdot \ba^\beta,
\end{equation*}
  the second fundamental form, given as curvature tensor, in covariant or mixed components, respectively, by
\begin{equation}\label{2.1c}
b_{\alpha\beta}:=\ba^3 \cdot \d_\beta \ba_\alpha, \qquad b_{\alpha}^\beta:=a^{\beta\sigma}\cdot b_{\sigma\alpha},
\end{equation}
and the Christoffel symbols of the surface $S$ as
\begin{equation}\label{2.1d}
\Gamma^\sigma_{\alpha\beta}:=\ba^\sigma\cdot \d_\beta \ba_\alpha.
\end{equation}
The area element along $S$ is $\sqrt{a}dy$ where
\begin{equation}\label{definicion_a}
a:=\det (a_{\alpha\beta}).
\end{equation}
Let $\gamma_0$ be a subset  of  $\gamma$, such that $\meas (\gamma_0)>0$.
For each $\varepsilon>0$, we define the three-dimensional domain $\Omega^\varepsilon:=\omega \times (-\varepsilon, \varepsilon)$ and  its boundary $\Gamae=\partial\Omega^\var$. We also define  the following parts of the boundary,
\begin{align*}
\Gamma^\varepsilon_+:=\omega\times \{\varepsilon\}, \quad \Gamma^\varepsilon_C:= \omega\times \{-\varepsilon\},\quad \Gamma_0^\varepsilon:=\gamma_0\times[-\varepsilon,\varepsilon].
\end{align*}

Let $\bx^\varepsilon=(x_i^\varepsilon)$ be a generic point of $\bar{\Omega}^\varepsilon$ and let $\d_i^\var$ denote the partial derivative with respect to $x_i^\varepsilon$. Note that $x_\alpha^\varepsilon=y_\alpha$ and $\d_\alpha^\varepsilon =\d_\alpha$. Let $\bTheta:\bar{\Omega}^\varepsilon\rightarrow \mathbb{R}^3$ be the mapping defined by
\begin{align} \label{bTheta}
\bTheta(\bx^\varepsilon):=\btheta(\by) + x_3^\varepsilon \ba_3(\by) \ \forall \bx^\varepsilon=(\by,x_3^\varepsilon)=(y_1,y_2,x_3^\varepsilon)\in\bar{\Omega}^\varepsilon.
\end{align}
The next theorem shows that if the injective mapping $\btheta:\bar{\omega}\rightarrow\mathbb{R}^3$ is smooth enough, the mapping $\bTheta:\bar{\Omega}^\var\rightarrow\mathbb{R}^3$ is also injective for $\var>0$ small enough (see Theorem 3.1-1, \cite{Ciarlet4b}) and the vectors $\bg_i^\varepsilon(\bx^\varepsilon):=\d_i^\varepsilon\bTheta(\bx^\varepsilon)$ are linearly independent.

In what follows, and for the sake of briefness, we shall omit the explicit dependence on the space variable when there is no ambiguity.

\begin{theorem}\label{var_0}
Let $\omega$ be a domain in $\mathbb{R}^2$. Let $\btheta\in\mathcal{C}^2(\bar{\omega};\mathbb{R}^3)$ be an injective mapping such that the two vectors $\ba_\alpha=\d_\alpha\btheta$ are linearly independent at all points of $\bar{\omega}$ and let $\ba_3$,  defined  in (\ref{a_3}). Then there exists $\var_0>0$ such that   the mapping $\bTheta:\bar{\Omega}_0 \rightarrow\mathbb{R}^3$ defined by
\begin{align*}
\bTheta(\by,x_3):=\btheta(\by) + x_3 \ba_3(\by) \ \forall (\by,x_3)\in\bar{\Omega}_0, \ \textrm{where} \ \Omega_0:=\omega\times(-\var_0,\var_0),
\end{align*}
is a $\mathcal{C}^1-$ diffeomorphism from $\bar{\Omega}_0$ onto $\bTheta(\bar{\Omega}_0)$ and $\det (\bg_1,\bg_2,\bg_3)>0$ in $\bar{\Omega}_0$, where $\bg_i:=\d_i\bTheta$.
\end{theorem}
For each $\var$, $0<\var\le\var_0$, the set $\bTheta(\bar{\Omega}^\var)$ is the reference configuration of an elastic shell, with middle surface $S=\btheta(\bar{\omega})$ and thickness $2\varepsilon>0$. By the theorem above, the mapping $\bTheta:\bar{\Omega}^\varepsilon\rightarrow \mathbb{R}^3$ is injective for all $\var$, $0<\var\le\var_0$ and, moreover, the three vectors $\bg_i^\varepsilon(\bx^\varepsilon)$ form the covariant basis at the point $\bTheta(\bx^\varepsilon)$, and $\bg^{i,\varepsilon}(\bx^\varepsilon) $ defined by the relations
\begin{equation}\label{gis}
\bg^{i,\varepsilon}\cdot \bg_j^\varepsilon=\delta_j^i,
\end{equation}
form the contravariant basis at the point $\bTheta(\bx^\varepsilon)$. The covariant and contravariant components of the metric tensor are defined, respectively, as
\begin{align*}
 g_{ij}^\varepsilon:=\bg_i^\varepsilon \cdot \bg_j^\varepsilon,\quad g^{ij,\varepsilon}:=\bg^{i,\varepsilon} \cdot \bg^{j,\varepsilon},
\end{align*}
and Christoffel symbols as
\begin{align} \label{simbolos3D}
\Gamma^{p,\varepsilon}_{ij}:=\bg^{p,\varepsilon}\cdot\d_i^\varepsilon \bg_j^\varepsilon. 
\end{align}

The volume element in the set $\bTheta(\bar{\Omega}^\varepsilon)$ is $\sqrt{g^\varepsilon}dx^\var$ and the surface element in $\bTheta(\Gamma^\varepsilon)$ is $\sqrt{g^\varepsilon}d\Gamae$  where
\begin{align} \label{g}
g^\varepsilon:=\det (g^\varepsilon_{ij}).
\end{align}
Let $\bn^\varepsilon(\bx^\varepsilon)$ denote the unit outward normal vector on $\bx^\varepsilon\in\Gamma^\var$ and  $\hat{\bn}^\varepsilon(\hat{\bx}^\var)$ the unit outward normal vector on $\hat{\bx}^\var=\bTheta(\bx^\varepsilon)\in\bTheta(\Gamma^\varepsilon)$. It is verified that (see, \cite[p. 41]{Ciarlet2})
$$
\hat{\bn}^\var(\hat{\bx}^\var)=\frac{\Cof(\nabla\bTheta(\bx^\var))\bn^\var(\bx^\var)}{|\Cof(\nabla\bTheta(\bx^\var))\bn^\var(\bx^\var)|}.
$$
We are particularly interested in the normal components of vectors on $\bTheta(\Gamma_C^\var)$. Recall that on $\Gamma_C^\var$, it is verified that $\bn^\var=(0,0,-1)$. Also, note that from (\ref{bTheta}) we deduce that $\bg_3^\var=\bg^{3,\var}=\ba_3$ and therefore $g^{33,\var}=|\bg^{3,\var}|=1$. These arguments imply that, in particular,
\begin{equation}\label{normal}
\hat{\bn}^\var(\hat{\bx}^\var)=-\bg_3(\bx^\var)=-\ba_3(\by),\ {\rm where}\ \hat{\bx}^\var=\bTheta(\bx^\var),\ {\rm and}\ \bx^\var=(\by,-\var)\in\Gamma_C^\var.
\end{equation}
Now, for a field $\hat{\bv}^\var$ defined in $\bTheta(\bar{\Omega}^\var)$, where the cartesian basis is denoted by $\{\hat{\be}^i\}_{i=1}^3$, we define its covariant curvilinear coordinates $(v_i^\var)$ in $\bar{\Omega}^\var$ as follows:
\begin{equation}\label{normal2}
\hat{\bv}^\var(\hat{\bx}^\var)=\hat{v}^\var_i(\hat{\bx}^\var)\hat{\be}^i=:v_i^\var(\bx^\var)\bg^{i,\var}(\bx^\var),\ {\rm with}\ \hat{\bx}^\var=\bTheta(\bx^\var).
\end{equation}
Therefore, by combining (\ref{gis}), (\ref{normal}) and (\ref{normal2}), it can be shown that, on $\Gamma_C^\var$, we have
$$
\hat{v}_n:=\hat{\bv}^\var\cdot\hat{\bn}^\var=(\hat{v}_i^\var\hat{n}^{i,\var})=(\hat{v}_i^\var\hat{\be}^i)\cdot(-\bg_3)%
=(v_i^\var\bg^{i,\var})\cdot(-\bg_3)=-v_3^\var.
$$
Also, since $v_i^\var n^{i,\var}=-v_3^\var$ on $\Gamma_C^\var$, it is verified in particular that
\begin{equation}\label{normal3}
\hat{v}_n=(\hat{v}_i^\var\hat{n}^{i,\var})=v_i^\var n^{i,\var}=-v_3^\var.
\end{equation}

We assume that $\bTheta(\bar{\Omega}^\varepsilon)$ is a natural state  of a shell made of an elastic material, which is homogeneous and isotropic, so that the material is characterized by its Lam\'e coefficients   $\lambda\geq0, \mu>0$. We assume that these constants are independent of $\var$.

We also assume that the shell is subjected to a boundary condition of place; in particular, the displacements field vanishes in  $\bTheta(\Gamma_0^\varepsilon)$,  this is,  a portion of the lateral face of the shell.

Further, under the effect of applied volumic forces of density $\hat{\bbf}^\var=(\hat{f}^{i,\var})$ acting in $\bTheta(\Omega^\var)$ and tractions of density $\hat{\bbh}^\var=(\hat{h}^{i,\var})$ acting upon $\bTheta(\Gamma^\var_+)$, the elastic shell is deformed and may enter in contact with a rigid foundation which, initially, is at a known distance $s^\varepsilon$ measured along the direction of $\hat{\bn}^\var$ on $\bTheta(\Gamma_C^\var)$. For simplicity, we take $s^\var=0$ in the following. Besides, in the following we shall use the shorter notation $\hat{\Omega}^\var=\bTheta(\Omega^\var)$, $\hat{\Gamma}^\var=\bTheta(\Gamma^\var)$. It is well known that the variational formulation of the unilateral contact problem in cartesian coordinates is the following:

\begin{problem}\label{problema_cart}
Find $\hat{\bu}^\varepsilon=(\hat{u}_i^\varepsilon):\hat{\Omega}^\varepsilon \rightarrow \mathbb{R}^3$  such that,
\begin{align*}
& \hat{\bu}^\varepsilon\in K(\hat{\Omega}^\var):=\{\hat{\bv}^\varepsilon=(\hat{v}_i^\varepsilon)\in [H^1(\hat{\Omega}^\var)]^3; \hat{\bv}^\varepsilon=\mathbf{\bcero} \ {\rm on} \ \hat{\Gamma}_0^\var; \hat{v}_n\le0\ {\rm on}\ \hat{\Gamma}_C^\var\},\\
&\int_{\hat{\Omega}^\var}\hat{A}^{ijkl,\varepsilon}\hat{e}^\varepsilon_{kl}(\hat{\bu}^\varepsilon)%
   (\hat{e}^\varepsilon_{ij}(\hat{\bv}^\varepsilon)-\hat{e}^\varepsilon_{ij}(\hat{\bu}^\varepsilon)) d\hat{x}^\varepsilon\\%
&\quad   \ge \int_{\hat{\Omega}^\var} \hat{f}^{i,\varepsilon} (\hat{v}_i^\varepsilon-\hat{u}_i^\varepsilon)\,d\hat{x}^\varepsilon %
   +\int_{\hat{\Gamma}_+^\var} \hat{h}^{i,\varepsilon} (\hat{v}_i^\varepsilon-\hat{u}_i^\varepsilon)\,d\hat{\Gamma}^\var  \quad \forall \hat{\bv}^\varepsilon\in K(\hat{\Omega}^\var),
  \end{align*}
\end{problem}
where
$$
\hat{A}^{ijkl,\varepsilon}=\lambda\delta^{ij}\delta^{kl}+\mu(\delta^{ik}\delta^{jl}+\delta^{il}\delta^{jk}),\quad
\hat{e}^\var_{ij}(\hat{\bv}^\var)=\frac{1}{2}(\hat{\partial}_j\hat{v}_i^\var+\hat{\partial}_i\hat{v}_j^\var),
$$
denote the elasticity fourth-order tensor and the deformation operator, respectively. We now define the corresponding contravariant components in curvilinear coordinates for the applied forces densities:
$$
\hat{f}^{i,\var}(\hat{\bx}^\var)\hat{\be}_i=:f^{i,\var}(\bx^\var)\bg_i^\var(\bx^\var),\quad%
\hat{h}^{i,\var}(\hat{\bx}^\var)\hat{\be}_id\hat{\Gamma}^\var=:h^{i,\var}(\bx^\var)\bg_i^\var(\bx^\var)\sqrt{g^\var(\bx^\var)}d\Gamma^\var, $$
and the covariant components in curvilinear coordinates for the displacements field:
$$
\hat{\bu}^\var(\hat{\bx}^\var)=\hat{u}^\var_i(\hat{\bx}^\var)\hat{\be}^i=:u_i^\var(\bx^\var)\bg^{i,\var}(\bx^\var),\ {\rm with}\ \hat{\bx}^\var=\bTheta(\bx^\var).
$$
Let us define the space,
\begin{align*} 
V(\Omega^\varepsilon)=\{\bv^\varepsilon=(v_i^\varepsilon)\in [H^1(\Omega^\varepsilon)]^3; \bv^\varepsilon=\mathbf{\bcero} \ {\rm on} \ \Gamma_0^\varepsilon\}.
\end{align*}
This is a real Hilbert space with the induced inner product of $[H^1(\Omega^\var)]^3$. The corresponding norm  is denoted by $||\cdot||_{1,\Omega^\var}$. By (\ref{normal3}), we deduce that the condition $\hat{v}_n\le0$ in the definition of $K(\hat{\Omega}^\var)$ in Problem \ref{problema_cart} is equivalent to $v_3^\var\ge0$. Therefore, let us define the following subset of admissible unknowns:
\begin{align*} 
K(\Omega^\varepsilon)=\{\bv^\varepsilon=(v_i^\varepsilon)\in V(\Omega^\varepsilon); v_3^\varepsilon\ge0 \ {\rm on} \ \Gamma_C^\varepsilon  \}.
\end{align*}
This is a non-empty, closed and convex subset of $V(\Omega^\varepsilon)$. With this definitions it is straightforward to derive from the Problem \ref{problema_cart} the following variational problem:
\begin{problem}\label{problema_eps}
Find $\bu^\varepsilon=(u_i^\varepsilon):{\Omega}^\varepsilon \rightarrow \mathbb{R}^3$  such that,
\begin{align*}
  \displaystyle
  & \bu^\varepsilon\in K(\Omega^\varepsilon),
  \quad\int_{\Omega^\varepsilon}A^{ijkl,\varepsilon}e^\varepsilon_{k||l}(\bu^\varepsilon)%
   (e^\varepsilon_{i||j}(\bv^\varepsilon)-e^\varepsilon_{i||j}(\bu^\varepsilon))\sqrt{g^\varepsilon} dx^\varepsilon\nonumber\\
 & \quad\ge \int_{\Omega^\varepsilon} f^{i,\varepsilon} (v_i^\varepsilon-u_i^\varepsilon) \sqrt{g^\varepsilon} dx^\varepsilon + \int_{\Gamma_+^\varepsilon} h^{i,\varepsilon} (v_i^\varepsilon-u_i^\varepsilon)\sqrt{g^\varepsilon}  d\Gamma^\varepsilon  \quad \forall \bv^\varepsilon\in K(\Omega^\varepsilon),
  \end{align*}
\end{problem}
where the functions $A^{ijkl,\varepsilon}=A^{jikl,\varepsilon}=A^{klij,\varepsilon}\in{\cal C}^1(\bar{\Omega}^\var)$, defined by
\begin{equation}\label{TensorAeps}
A^{ijkl,\varepsilon}:= \lambda g^{ij,\varepsilon}g^{kl,\varepsilon} + \mu(g^{ik,\varepsilon}g^{jl,\varepsilon} + g^{il,\varepsilon}g^{jk,\varepsilon} ), 
\end{equation}
represent the contravariant components of the three-dimensional elasticity tensor, and the functions $e^\varepsilon_{i||j}(\bv^\var)=e^\varepsilon_{j||i}(\bv^\var)\in L^2(\Omega^\var)$ are defined for all $\bv^\var\in [H^1(\Omega^\var)]^3$ by
\begin{align*}
e^\varepsilon_{i||j}(\bv^\var):= \frac1{2}(v^\varepsilon_{i||j}+ v^\varepsilon_{j||i})=\frac1{2}(\d^\varepsilon_jv^\varepsilon_i + \d^\varepsilon_iv^\varepsilon_j) - \Gamma^{p,\varepsilon}_{ij}v^\varepsilon_p.
\end{align*}
Note that the following additional relations are satisfied,
$$
\Gamma^{3,\varepsilon}_{\alpha 3}=\Gamma^{p,\varepsilon}_{33}=0  \ \textrm{in} \ \bar{\Omega}^\varepsilon,%
\quad  A^{\alpha\beta\sigma 3,\varepsilon}=A^{\alpha 333,\varepsilon}=0 \ \textrm{in} \ \bar{\Omega}^\varepsilon,
$$
as a consequence of the definition of $\bTheta$ in (\ref{bTheta}). The definition of the fourth order tensor (\ref{TensorAeps}) imply that (see Theorem 1.8-1, \cite{Ciarlet4b}) for $\var>0$ small enough, there exists a constant $C_e>0$, independent of $\var$, such that,
\begin{align} \label{elipticidadA}
 \sum_{i,j}|t_{ij}|^2\leq C_e A^{ijkl,\var}(\bx^\var)t_{kl}t_{ij},
\end{align}
for all $\bx^\var\in\bar{\Omega}^\var$ and all $\bt=(t_{ij})\in\mathbb{S}^2$.
{ The unique solvability of Problem \ref{problema_eps} for $\var>0$ small enough is straightforward. Indeed, combining the use of (\ref{elipticidadA}) and a Korn inequality (see for example \cite[Th. 1.7-4]{Ciarlet4b})}, and given that $K(\Omega^\var)$ is a closed non-empty convex set, we can cast this problem in the framework of the elliptic variational inequalities theory (see, for example {\cite{Capelo,Jarusek,HS}}), and conclude the existence and uniqueness of $\bu^\varepsilon\in K(\Omega^\varepsilon)$, solution of Problem \ref{problema_eps}.

\begin{remark}
We recall that the vector field $\bu^\varepsilon=(u_i^\varepsilon):{\Omega}^\varepsilon \rightarrow \mathbb{R}^3$ solution of Problem \ref{problema_eps} has to be interpreted conveniently. The functions  $u_i^\varepsilon:\bar{\Omega}^\varepsilon \rightarrow \mathbb{R}^3$ are the covariant components of the ``true" displacements field $\bUcal^\var:=u_i^\varepsilon \bg^{i,\varepsilon}:\bar{\Omega}^\varepsilon \rightarrow \mathbb{R}^3$.
\end{remark}

If the functions involved have sufficient regularity, from Problem \ref{problema_eps} we can deduce the following strong formulation:

\begin{problem}\label{problema_mecanico}
Find $\bu^\var=(u_i^\var):\Omega^\var\longrightarrow \mathbb{R}^3$ such that,
\begin{align}\label{equilibrio}
&-\sigma^{ij,\var}||_j(\bu^\var)=f^{i,\var} \en \Omega^\var, \\\label{Dirichlet}
&u_i^\var=0 \on \Gamae_0, \\\label{Neumann}
&\sigma^{ij,\var}(\bu^\var)n_j^\var=h^{i,\var} \on \Gamma^\var_+,\\
&u_3^\var\ge 0,\ \sigma^{33,\var}(\bu^\var)\le 0,\ \sigma^{33,\var}(\bu^\var)u_3^\var=0,\ {\sigma^{3\alpha,\var}(\bu^\var)=0}\  \on \Gamma^\var_C,\label{Signorini}
\end{align}
where the functions
\begin{align*}
\sigma^{ij,\var}(\bu^\var):=A^{ijkl,\var}\ekl^\var(\bu^\var),
\end{align*}
are the contravariant components of the linearized stress tensor field and the functions
\begin{align*}
\sigma^{ij,\var}||_k(\bu^\var):= \d_k^\var\sigma^{ij,\var}(\bu^\var) + \Gamma_{pk}^{i,\var} \sigma^{pj,\var}(\bu^\var) + \Gamma_{kq}^{j,\var}\sigma^{iq,\var}(\bu^\var),
\end{align*}
denote the first-order covariant derivatives of the stress tensor components.
\end{problem}

We now proceed to describe the equations in Problem \ref{problema_mecanico}. Expression (\ref{equilibrio}) is the equilibrium equation. The equality (\ref{Dirichlet}) is the Dirichlet condition of place, (\ref{Neumann}) is the Neumann condition and (\ref{Signorini}) is the Signorini condition of unilateral, frictionless, contact. 

\section{The scaled three-dimensional shell problem}\label{seccion_dominio_ind}

For convenience, we consider a reference domain independent of the small parameter $\var$. Hence, let us define the three-dimensional domain $\Omega:=\omega \times (-1, 1) $ and  its boundary $\Gamma=\partial\Omega$. We also define the following parts of the boundary,
 \begin{align*}
 \Gamma_+:=\omega\times \{1\}, \quad \Gamma_C:= \omega\times \{-1\},\quad \Gamma_0:=\gamma_0\times[-1,1].
 \end{align*}
 Let $\bx=(x_1,x_2,x_3)$ be a generic point in $\bar{\Omega}$ and we consider the notation $\d_i$ for the partial derivative with respect to $x_i$. We define the  projection map $\pi^\varepsilon: \bar{\Omega} \longrightarrow\bar{\Omega}^\varepsilon,$ such that
 \begin{align*}
 \pi^\varepsilon(\bx)=\bx^\varepsilon=(x_i^\varepsilon)=(x_1^\var,x_2^\var,x_3^\var)=(x_1,x_2,\varepsilon x_3)\in \bar{\Omega}^\varepsilon,
 \end{align*}
 hence, $\d_\alpha^\varepsilon=\d_\alpha $  and $\d_3^\varepsilon=\frac1{\varepsilon}\d_3$. We consider the scaled unknown $\bu(\varepsilon)=(u_i(\varepsilon)):\bar{\Omega}\longrightarrow \mathbb{R}^3$ and the scaled vector fields $\bv=(v_i):\bar{\Omega}\longrightarrow \mathbb{R}^3 $ defined as
 \begin{align*}
 u_i^\varepsilon(\bx^\varepsilon)=:u_i(\varepsilon)(\bx), \ \textrm{and} \ v_i^\varepsilon(\bx^\varepsilon)=:v_i(\bx) \ \forall \bx\in\bar{\Omega},\ \bx^\varepsilon=\pi^\varepsilon(\bx)\in \bar{\Omega}^\varepsilon.
 \end{align*}
We remind that, by hypothesis, the Lam\'e constants are independent of $\varepsilon$. Also, let the functions, $\Gamma_{ij}^{p,\varepsilon}, g^\varepsilon, A^{ijkl,\varepsilon}$ defined in (\ref{simbolos3D}), (\ref{g}), (\ref{TensorAeps}) be associated with the functions $\Gamma_{ij}^p(\varepsilon),$ $ g(\varepsilon),$ $ A^{ijkl}(\varepsilon),$ defined by
  \begin{align} \label{escalado_simbolos}
  &\Gamma_{ij}^p(\varepsilon)(\bx):=\Gamma_{ij}^{p,\varepsilon}(\bx^\varepsilon),\\\label{escalado_g}
  & g(\varepsilon)(\bx):=g^\varepsilon(\bx^\varepsilon),\\\label{tensorA_escalado}
  & A^{ijkl}(\varepsilon)(\bx):=A^{ijkl,\varepsilon}(\bx^\varepsilon),
  \end{align}
for all $\bx\in\bar{\Omega}$, $\bx^\varepsilon=\pi^\varepsilon(\bx)\in\bar{\Omega}^\varepsilon$. For all $\bv=(v_i)\in [H^1(\Omega)]^3$, let there be associated the scaled linearized strains $(\eij(\var)(\bv))\in L^2(\Omega)$, which we also denote as $(\eij(\var;\bv))$, defined by
\begin{align*} 
&\eab(\varepsilon;\bv):=\frac{1}{2}(\d_\beta v_\alpha + \d_\alpha v_\beta) - \Gamma_{\alpha\beta}^p(\varepsilon)v_p,\\ 
  & \eatres(\varepsilon;\bv):=\frac{1}{2}(\frac{1}{\var}\d_3 v_\alpha + \d_\alpha v_3) - \Gamma_{\alpha 3}^p(\varepsilon)v_p,\\ 
  & \edtres(\varepsilon;\bv):=\frac1{\varepsilon}\d_3v_3.
\end{align*}
Note that with these definitions it is verified that
\begin{align*}
\eij^\var(\bv^\var)(\pi^\var(\bx))=\eij(\var;\bv)(\bx) \ \forall\bx\in\Omega.
\end{align*}

 \begin{remark} The functions $\Gamma_{ij}^p(\varepsilon), g(\varepsilon), A^{ijkl}(\varepsilon)$ converge in $\mathcal{C}^0(\bar{\Omega})$ when $\varepsilon$ tends to zero.
 \end{remark}

 \begin{remark}When we consider
 $\varepsilon=0$ the functions will be defined with respect to $\by\in\bar{\omega}$. We shall distinguish the three-dimensional Christoffel symbols from the two-dimensional ones associated to $S$ by using  $\Gamma_{\alpha \beta}^\sigma(\varepsilon)$ and $ \Gamma_{\alpha\beta}^\sigma$, respectively.
 \end{remark}
We will study the asymptotic behavior of the scaled contravariant components $A^{ijkl}(\var)$ of the three-dimensional elasticity tensor defined in (\ref{tensorA_escalado}), as $\var\rightarrow0$.  We show the uniform positive definiteness  not only with respect to $\bx\in\bar{\Omega}$, but also with respect to $\var$, $0<\var\leq\var_0$. Furthermore, the limits are functions of $\by\in\bar{\omega}$ only, that is, independent of the transversal variable $x_3$. Let us recall Theorem 3.3-2, \cite{Ciarlet4b}.

\begin{theorem} \label{Th_comportamiento asintotico}
Let $\omega$  be a domain in $\mathbb{R}^2$, $\btheta\in\mathcal{C}^2(\bar{\omega};\mathbb{R}^3)$ be an injective mapping such that the two vectors $\ba_\alpha=\d_\alpha\btheta$ are linearly independent at all points of $\bar{\omega}$ and $a^{\alpha\beta}$ denote the contravariant components of the metric tensor of $S=\btheta(\bar{\omega})$. In addition to that, let the other assumptions on the mapping $\btheta$ and the definition of $\var_0$ be as in Theorem \ref{var_0}. The contravariant components $A^{ijkl}(\var)$ of the scaled three-dimensional elasticity tensor defined in (\ref{tensorA_escalado}) satisfy
\begin{align*}
A^{ijkl}(\var)= A^{ijkl}(0) + O(\var) \ \textrm{and} \ A^{\alpha\beta\sigma 3}(\var)=A^{\alpha 3 3 3}(\var)=0,
\end{align*}
for all $\var$, $0<\var \leq \var_0$, and
\begin{align*}
A^{\alpha\beta\sigma\tau}(0)&= \lambda a^{\alpha\beta}a^{\sigma\tau} + \mu(a^{\alpha\sigma}a^{\beta\tau} + a^{\alpha\tau}a^{\beta\sigma}), & A^{\alpha\beta 3 3}(0)&= \lambda a^{\alpha\beta},
\\
 A^{\alpha 3\sigma 3}(0)&=\mu a^{\alpha\sigma} ,& A^{33 3 3}(0)&= \lambda + 2\mu,
 \\
A^{\alpha\beta\sigma 3}(0) &=A^{\alpha 333}(0)=0.
\end{align*}
Moreover, there exists a constant $C_e>0$, independent of the variables and $\var$, such that
  \begin{align} \label{elipticidadA_eps}
  \sum_{i,j}|t_{ij}|^2\leq C_e A^{ijkl}(\varepsilon)(\bx)t_{kl}t_{ij},
  \end{align}
 for all $\var$, $0<\var\leq\var_0$, for all $\bx\in\bar{\Omega}$ and all $\bt=(t_{ij})\in\mathbb{S}^2$.
\end{theorem}

\begin{remark}
The asymptotic behavior of $g(\var)$ and the contravariant components of the elasticity tensor, $A^{ijkl}(\var)$, also implies that
\begin{align} \label{tensorA_tildes}
A^{ijkl}(\var)\sqrt{g(\var)}= A^{ijkl}(0)\sqrt{a} + \var \tilde{A}^{ijkl,1} + \var^2 \tilde{A}^{ijkl,2} + o(\var^2),
\end{align}
for certain regular contravariant components $\tilde{A}^{ijkl,\alpha}$ of certain tensors.
\end{remark}
Let the scaled applied forces $\bbf(\varepsilon):\Omega\longrightarrow \mathbb{R}^3$ and  $\bbh(\varepsilon):\Gamma_+\longrightarrow \mathbb{R}^3$ be defined by
   \begin{align*}
  \bbf^\var&=(f^{i,\varepsilon})(\bx^\varepsilon)=:\bbf(\var)= (f^i(\varepsilon))(\bx)
  \quad \forall \bx\in\Omega, \ \textrm{where} \ \bx^\varepsilon=\pi^\varepsilon(\bx)\in \Omega^\varepsilon, \\
   \bbh^\var&=(h^{i,\varepsilon})(\bx^\varepsilon)=:\bbh(\var)= (h^i(\varepsilon))(\bx)
   \quad \forall \bx\in\Gamma_+, \ \textrm{where} \ \bx^\varepsilon=\pi^\varepsilon(\bx)\in \Gamma_+^\varepsilon .
\end{align*}
Also, we define the space
\begin{align*} 
V(\Omega)=\{\bv=(v_i)\in [H^1(\Omega)]^3; \bv=\mathbf{0} \ on \ \Gamma_0\},
\end{align*}
which is a Hilbert space, with associated norm denoted by $||\cdot||_{1,\Omega}$. We also define the non-empty closed convex subset
$$
K(\Omega)=\{\bv=(v_i)\in V(\Omega); v_3\ge0\ on \ \Gamma_C\}.
$$
The scaled variational problem  can then  be written as follows:
\begin{problem}\label{problema_escalado}
Find $\bu(\varepsilon):\Omega\longrightarrow \mathbb{R}^3$ such that,
\begin{align}\nonumber
&\bu(\varepsilon)\in K(\Omega),\nonumber \\
&\int_{\Omega}A^{ijkl}(\varepsilon)e_{k||l}(\varepsilon;\bu(\varepsilon))%
(e_{i||j}(\varepsilon;\bv)-e_{i||j}(\varepsilon;\bu(\varepsilon)))\sqrt{g(\varepsilon)} dx \nonumber \\
& \quad\ge \int_{\Omega} f^{i}(\varepsilon) (v_i-u_i(\varepsilon)) \sqrt{g(\varepsilon)} dx
+\frac{1}{\varepsilon}\int_{\Gamma_+} h^{i}(\varepsilon) (v_i-u_i(\varepsilon))\sqrt{g(\varepsilon)}  d\Gamma  \quad \forall \bv\in K(\Omega), \label{ec_problema_escalado}
\end{align}
\end{problem}

\begin{remark}
Note that the order of the applied forces has not been determined yet.
\end{remark}
{ The unique solvability of Problem \ref{problema_escalado} for $\var>0$ small enough is straightforward. Indeed, combining the use of (\ref{elipticidadA_eps}) and a Korn inequality (see for example \cite[Th. 1.7-4]{Ciarlet4b})}, and given that $K(\Omega)$ is a closed non-empty convex set, we can cast this problem in the framework of the elliptic variational inequalities theory (see, for example {\cite{Capelo,Jarusek,HS}}), and conclude the existence and uniqueness of $\bu(\varepsilon)\in K(\Omega)$, solution of Problem \ref{problema_escalado}.


We now present some additional results which will be used in the next section. First, we recall the Theorem 3.3-1, \cite{Ciarlet4b}.

\begin{theorem} \label{Th_simbolos2D_3D}
Let $\omega$ be a domain in $\mathbb{R}^2$, let $\btheta\in\mathcal{C}^3(\bar{\omega};\mathbb{R}^3)$ be an injective mapping such that the two vectors $\ba_\alpha=\d_\alpha\btheta$ are linearly independent at all points of $\bar{\omega}$ and let $\var_0>0$ be as in Theorem \ref{var_0}. The functions $\Gamma^p_{ij}(\var)=\Gamma^p_{ji}(\var)$ and $g(\var)$ are defined in (\ref{escalado_simbolos})--(\ref{escalado_g}), the functions $b_{\alpha\beta}, b_\alpha^\sigma, \Gamma_{\alpha\beta}^\sigma$ and $a$ are defined in (\ref{2.1c})--(\ref{definicion_a}) and the covariant derivatives $b_\beta^\sigma|_\alpha$ are defined by
\begin{align} \label{b_barra}
b_\beta^\sigma|_\alpha:=\d_\alpha b_\beta^\sigma +\Gamma^\sigma_{\alpha\tau}b_\beta^\tau - \Gamma^\tau_{\alpha\beta}b^\sigma_\tau.
\end{align}
The functions $b_{\alpha\beta}, b_\alpha^\sigma, \Gamma_{\alpha\beta}^\sigma, b_\beta^\sigma|_\alpha$ and $a$ are identified with functions in $\mathcal{C}^0(\bar{\Omega})$. Then
\begin{align*}
\begin{aligned}[c]
 \Gamma_{\alpha\beta}^\sigma(\var)&=  \Gamma_{\alpha\beta}^\sigma -\var x_3b_\beta^\sigma|_\alpha + O(\var^2), \\
  \d_3 \Gamma_{\alpha\beta}^p(\var)&= O(\var),
   \\
   \Gamma_{\alpha3}^3(\var)&=\Gamma_{33}^p(\var)=0,
\end{aligned}
\qquad
\begin{aligned}[c]
 \Gamma_{\alpha\beta}^3(\var)&=b_{\alpha\beta} - \var x_3 b_\alpha^\sigma b_{\sigma\beta},
 \\
 \Gamma_{\alpha3}^\sigma(\var)& = -b_\alpha^\sigma - \var x_3 b_\alpha^\tau b_\tau^\sigma + O(\var^2),
\\
 g(\varepsilon)&=a + O(\varepsilon),
\end{aligned}
\end{align*}
for all $\var$, $0<\var\leq\var_0$, where the order symbols $O(\var)$ and $O(\var^2)$  are meant with respect to the norm $||\cdot||_{0,\infty,\bar{\Omega}}$ defined by
\begin{align*} 
||w||_{0,\infty,\bar{\Omega}}=\sup \{|w(\bx)|; \bx\in\bar{\Omega}\}.
\end{align*}
 Finally, there exist constants $a_0, g_0$ and $g_1$ such that
 \begin{align*}
 & 0<a_0\leq a(\by) \ \forall \by\in \bar{\omega},
 \\ 
 & 0<g_0\leq g(\varepsilon)(\bx) \leq g_1 \ \forall \bx\in\bar{\Omega} \ \textrm{and} \ \forall \ \var, 0<\varepsilon\leq \varepsilon_0.
 \end{align*}
\end{theorem}
In \cite[Theorem 3.4-1]{Ciarlet4b}, we find the following useful result
\begin{theorem}\label{th_int_nula}
 Let $\omega$ be a domain in $\mathbb{R}^2$ with boundary $\gamma$, let $\Omega=\omega\times (-1,1)$, and let $g\in L^p(\Omega)$, $p>1$, be a function such that
 \begin{align*}
 \intO g \d_3v dx=0, \ \textrm{for all} \ v\in \mathcal{C}^{\infty}(\bar{\Omega}) \ \textrm{with} \ v=0 \on \gamma\times[-1,1].
 \end{align*}
 Then $g=0$ a.e in $\Omega$.
\end{theorem}

In this work we also need the following extension for inequalities:

\begin{theorem} \label{th_int_post}
Let $\omega$ be a domain in $\mathbb{R}^2$ with boundary $\gamma$, let $\Omega=\omega\times (-1,1)$, and let $g\in L^p(\Omega)$, $p>1$, be a function such that
\begin{equation}\label{xx}
\intO g \d_3v dx\ge0, \ \textrm{for all} \ v\in \mathcal{C}^{\infty}(\bar{\Omega}) \ \textrm{with} \ v=0\ \textrm{on}\ \gamma\times[-1,1]\ \textrm{and}\ v\ge0\ \textrm{in}\ \Omega.
\end{equation}
Then { $g=0$} a.e. in $\Omega$.
\end{theorem}
\begin{proof}
Given $\varphi\in D(\Omega)$, with $\varphi\ge0$, we define
$$
v(x_1,x_2,x_3)=\int_{-1}^{x_3}\varphi(x_1,x_2,t)dt.
$$
Then, $v=0$ on $\gamma\times[-1,1]$ and $v\ge0$ in $\Omega$. Moreover,
$$
0\le\int_\Omega g\partial_3 v dx=\int_\Omega g\varphi dx,
$$
and this for all $\varphi\in D(\Omega)$, with $\varphi\ge0$. Therefore, $g\ge0$ a.e. in $\Omega$. Otherwise, let $E\subset\Omega$, with $\meas(E)>0$ where $g<0$. We can define $\varphi$ as a convenient regularization of a cut of the characteristic function of $E$, and therefore,
$$
\int_\Omega g\varphi dx=\int_{E}g\varphi dx<0,
$$
which is a contradiction with the previous statement.
{ Now, let $\varphi\in D(\omega)$, $\varphi\ge0$ and $v=e^{-x_3}\varphi$. Then, $v\ge0$ and $\partial_3v=-e^{-x_3}\varphi\le0$. Therefore, going back to (\ref{xx}) for this particular choice of $v$, we find that
$$
-\intO  g e^{-x_3}\varphi\sqrt{a}dx\ge0,
$$
which, having previously shown that $g\ge0$,  is only possible if $g=0$ a.e. in $\Omega$.}
\end{proof}
\begin{remark}
This result holds if $\intO g \d_3v dx\ge0$ for all $v\in H^1(\Omega)$ such that $v=0$ on $\Gamma_0$ and $v\ge0$ a.e. in $\Omega$. It is in this way that we will use this result in the following.
\end{remark}

\section{Formal Asymptotic Analysis} \label{procedure} 

In this section, we highlight some relevant steps in the construction of the formal asymptotic expansion of the scaled unknown variable $\bu(\var)$ including the characterization of the zeroth-order term, and the derivation of some key results which will lead to the two-dimensional variational inequalities of the elastic shell contact problems. We define the scaled applied forces as,
$$ 
\bbf(\varepsilon)(\bx)=\varepsilon^p\bbf^p(\bx) \ \forall \bx\in \Omega, \qquad 
\bbh(\varepsilon)(\bx)=\varepsilon^{p+1}\bbh^{p+1}(\bx) \ \forall  \bx\in \Gamma_+ ,
$$
{ where $\bbf^p$ and $\bbh^{p+1}$ are independent of $\var$, and $p$ is a natural number that will indicate the order of applied forces}. We substitute in  (\ref{ec_problema_escalado}) to obtain the following problem:
\begin{problem}\label{problema_orden_fuerzas}
Find $\bu(\varepsilon):\Omega\longrightarrow \mathbb{R}^3$ such that,
\begin{align}
&\bu(\varepsilon)\in K(\Omega),\quad
\int_{\Omega}A^{ijkl}(\varepsilon)e_{k||l}(\varepsilon;\bu(\varepsilon))%
(e_{i||j}(\varepsilon;\bv)-e_{i||j}(\varepsilon;\bu(\varepsilon)))\sqrt{g(\varepsilon)} dx \nonumber\\
&\quad\ge \int_{\Omega} \var^p \fb^{i,p}(v_i-u_i(\varepsilon)) \sqrt{g(\varepsilon)} dx
+\int_{\Gamma_+} \var^{p}\bh^{i,p+1}(v_i-u_i(\varepsilon))\sqrt{g(\varepsilon)}  d\Gamma  \quad \forall \bv\in K(\Omega).\label{ecuacion_orden_fuerzas}
\end{align}
\end{problem}

\begin{remark}
The existence and uniqueness of solution of Problem \ref{problema_orden_fuerzas} follows using analogous arguments as those used for Problem \ref{problema_escalado}.
\end{remark}
Assume that $\btheta\in\mathcal{C}^3(\bar{\omega};\mathbb{R}^3)$ and that the scaled unknown $\bu(\varepsilon)$ admits an asymptotic expansion of the form
\begin{align}\label{desarrollo_asintotico}
\bu(\varepsilon)&= \bu^0 + \varepsilon \bu^1 + \varepsilon^2 \bu^2 +... \quad \textrm{with} \ \bu^0\neq \mathbf{0},  \end{align}
where $\bu^0\in K(\Omega),$ $ \bu^q\in [H^1(\Omega)]^3$, $q\ge1$. The assumption (\ref{desarrollo_asintotico}) implies an asymptotic expansion of the scaled linear strain as follows
\begin{align*}
\eij(\var)\equiv\eij(\varepsilon;\bu(\varepsilon))&=\frac1{\varepsilon}\eij^{-1}+ \eij^0 + \varepsilon\eij^1 + \varepsilon^2\eij^2 + \varepsilon^3\eij^3+...
\end{align*}
where,
\begin{align*}
\left\{\begin{aligned}[c]
\eab^{-1}&=0, \\
  \eatres^{-1}&=\frac{1}{2}\d_3u_\alpha^0,
   \\
  \edtres^{-1}&=\d_3u_3^0,
\end{aligned}\right.
\qquad \qquad \qquad
\left\{\begin{aligned}[c]
 \eab^0&=\frac{1}{2}(\d_\beta u_\alpha^0 + \d_\alpha u_\beta^0) - \Gamma_{\alpha\beta}^\sigma u_\sigma^0 - b_{\alpha\beta}u_3^0,
 \\
\eatres^0&=\frac{1}{2}(\d_3 u_\alpha^1 + \d_\alpha u_3^0) +   b_{\alpha}^\sigma u_\sigma^0,
\\
\edtres^0&=\d_3u_3^1,
\end{aligned}\right.\qquad
\end{align*}
\begin{equation}\label{eij_terminos_expansion_u}
\end{equation}
\begin{align*}
\left\{\begin{aligned}[c]
\eab^1&=\frac{1}{2}(\d_\beta u_\alpha^1 + \d_\alpha u_\beta^1) - \Gamma_{\alpha\beta}^\sigma u_\sigma^1 - b_{\alpha\beta}u_3^1 + x_3(b_{\beta|\alpha}^\sigma u_\sigma^0  + b_\alpha^\sigma b_{\sigma\beta}u_3^0), \\
\eatres^1&=\frac{1}{2}(\d_3 u_\alpha^2 + \d_\alpha u_3^1) +   b_{\alpha}^\sigma u_\sigma^1 + x_3b_\alpha^\tau  b_\tau^\sigma u_\sigma^0, \\
\edtres^1&=\d_3 u_3^2.
\end{aligned}\right. \qquad \quad \qquad
\end{align*}
In addition, the functions $\eij(\varepsilon;\bv) $ admit the following expansion,
\begin{align*}
\eij(\varepsilon;\bv)=\frac{1}{\varepsilon}\eij^{-1}(\bv) + \eij^0(\bv) + \varepsilon\eij^1(\bv)+...,
\end{align*}
where,
\begin{align*}
\left\{\begin{aligned}[c]
\eab^{-1}(\bv)&=0,\\
 \eatres^{-1}(\bv)&=\frac{1}{2}\d_3v_\alpha,
   \\
\edtres^{-1}(\bv)&=\d_3v_3,
\end{aligned}\right.
\qquad \qquad \quad
\left\{\begin{aligned}[c]
 \eab^0(\bv)&=\frac{1}{2}(\d_\beta v_\alpha + \d_\alpha v_\beta) - \Gamma_{\alpha\beta}^\sigma v_\sigma - b_{\alpha\beta}v_3,
 \\
\eatres^0(\bv)&=\frac{1}{2} \d_\alpha v_3 +   b_{\alpha}^\sigma v_\sigma,
\\
\edtres^0(\bv)&=0,
\end{aligned}\right.
\end{align*}
\begin{equation}\label{eij_terminos_expansion}
\end{equation}
\begin{align*}
\left\{\begin{aligned}[c]
\eab^1(\bv)&=  x_3b_{\beta|\alpha}^\sigma v_\sigma  + x_3b_\alpha^\sigma b_{\sigma\beta}v_3, \\\nonumber
\eatres^1(\bv)&= x_3b_\alpha^\tau b_\tau^\sigma v_\sigma, \\\nonumber
\edtres^1(\bv)&=0.
\end{aligned}\right. \qquad \quad \qquad \qquad \qquad \qquad \qquad \qquad \qquad \qquad
\end{align*}
Upon substitution on (\ref{ecuacion_orden_fuerzas}), we proceed to characterize the terms involved in the asymptotic expansions considering different values for $p$, that is, taking different orders for the applied forces.

We shall now identify the leading term $\bu^0$ of the expansion (\ref{desarrollo_asintotico}) by canceling the other terms of the successive powers of $\var$ in the inequalities of the Problem \ref{problema_orden_fuerzas}. We will show that $\bu^0$ is solution of a variational formulation of a two-dimensional contact problem of an elastic membrane or flexural shell, depending on several factors, and that the orders of applied forces are determined in both cases. Given $\beeta=(\eta_i)\in [H^1(\omega)]^3,$ let
 \begin{equation} \label{def_gab}
 \gab(\beeta):= \frac{1}{2}(\d_\beta\eta_\alpha + \d_\alpha\eta_\beta) - \Gamma_{\alpha\beta}^\sigma\eta_\sigma -  b_{\alpha\beta}\eta_3,
 \end{equation}
 denote the covariant components of the linearized change of metric tensor associated with a displacement field $\eta_i\ba^i$ of the surface $S$. Let us define the spaces,
 \begin{align}\nonumber
 V(\omega)&:=\{\beeta=(\eta_i)\in[H^1(\omega)]^3 ; \eta_i=0 \ \textrm{on} \ \gamma_0 \}, \\ \nonumber
 V_0(\omega)&:=\{\beeta=(\eta_i)\in V(\omega), \gamma_{\alpha\beta}(\beeta)=0  \ \textrm{in} \ \omega \}, \\ \nonumber
 V_F(\omega)&:= \{ \beeta=(\eta_i) \in H^1(\omega)\times H^1(\omega)\times H^2(\omega) ;\ \eta_i=\d_\nu \eta_3=0 \on \gamma_0,\ \gab(\beeta)=0 \en \omega \},
 \end{align}
where $\d_\nu$ stands for the outer normal derivative along the boundary. We also define the sets
 \begin{align*}
 K(\omega)&:=\{\beeta=(\eta_i)\in V(\omega); \eta_3\ge0 \ \textrm{in} \ \omega \},\nonumber\\
 K_0(\omega)&:=V_0(\omega)\cap K(\omega),\nonumber\\
 K_F(\omega)&:=V_F(\omega)\cap K(\omega).\nonumber
 \end{align*}




\begin{theorem}\label{teo_parte1}
The main leading term $\bu^0$ of the asymptotic expansion (\ref{desarrollo_asintotico}) is independent of the transversal variable $x_3$. Therefore, it can be identified with a function $\bxi^0\in [H^1(\omega)]^3$ such that $\bxi^0=\bcero$ on $\gamma_0$. Moreover,
\begin{equation}\label{5.14b}
\bxi^0\in K(\omega).
\end{equation}
\end{theorem}
\begin{proof}
Let $p=-2$ in (\ref{ecuacion_orden_fuerzas}). Hence, grouping the terms multiplied by $\var^{-2}$ (see (\ref{tensorA_tildes})) we find that
\begin{align}
&\int_{\Omega} A^{ijkl}(0)\ekl^{-1}(\eij^{-1}(\bv)-\eij^{-1}) \sqrt{a}dx
\ge\int_{\Omega } f^{i,-2} (v_i-u_i^0) \sqrt{a} dx + \intG h^{i,-1} (v_i-u_i^0) \sqrt{a}d\Gamma.\label{et2}
\end{align}
Considering $\bv\in K(\Omega)$ independent of $x_3$ (see  (\ref{eij_terminos_expansion})),  the left-hand side of the  inequality (\ref{et2}) is non-positive, and then
$$
0\ge\int_{\Omega } f^{i,-2} (v_i-u_i^0) \sqrt{a} dx + \intG h^{i,-1} (v_i-u_i^0) \sqrt{a}d\Gamma.
$$
This would imply unwanted compatibility conditions between the applied forces. To avoid it, we must take $f^{i,-2}=0$ and $h^{i,-1}=0$. So that, back on  the inequality (\ref{et2}), using (\ref{eij_terminos_expansion_u}), (\ref{eij_terminos_expansion}) and Theorem \ref{Th_comportamiento asintotico}, leads to (note that $\eij^{-1}=\eij^{-1}(\bu^0)$),
\begin{align*}\nonumber
&0\le\intO A^{ijkl}(0)\ekl^{-1}\eij^{-1}(\bv-\bu^0)\sqrt{a} dx\\ %
&\quad= \intO\left(4A^{\alpha 3 \sigma 3}(0)\estres^{-1}\eatres^{-1}(\bv-\bu^0) %
+ A^{3333}(0)\edtres^{-1} \edtres^{-1}(\bv-\bu^0)  \right) \sqrt{a}dx\nonumber\\
&\quad = \intO \left( \mu a^{\alpha\sigma} \d_3u_\sigma^0 \d_3(v_\alpha-u_\alpha^0)%
+(\lambda + 2\mu) \d_3u_3^0 \d_3 (v_3-u_3^0) \right) \sqrt{a} dx,
\end{align*}
for all $\bv=(v_i)\in K(\Omega)$. By taking alternatively $\bv=\bzero$ and $\bv=2\bu^0$, we obtain the equality
$$
\intO \left( \mu a^{\alpha\sigma} \d_3u_\sigma^0 \d_3u_\alpha^0%
+(\lambda + 2\mu) \d_3u_3^0 \d_3 u_3^0 \right) \sqrt{a} dx=0.
$$
From here, similar arguments to those used in \cite[p. 167]{Ciarlet4b} show that
\begin{equation*}
\d_3 u_i^0=0.
\end{equation*}
Therefore, we have found that the main term $\bu^0$ of the asymptotic expansion is independent of the transversal variable, hence, it can be identified with a function $\bxi^0\in [H^1(\omega)]^3$ such that $\bxi^0=\bcero$ on $\gamma_0$, and $\xi_3^0\ge0$, this is, $\bxi^0\in K(\omega)$.
\end{proof}


Now, let us denote by $\a$ the contravariant components of the fourth order two-dimensional tensor defined as follows:
\begin{align} \label{tensor_a_bidimensional}
   \a&:=\frac{4\lambda\mu}{\lambda + 2\mu}a^{\alpha\beta}a^{\sigma\tau} + 2\mu \ten.
\end{align}

We recall the following result (see \cite[Theorem 3.3-2]{Ciarlet4b}) regarding the ellipticity of this tensor.

\begin{theorem}
 Let $\omega$  be a domain in $\mathbb{R}^2$, let $\btheta\in\mathcal{C}^1(\bar{\omega};\mathbb{R}^3)$ be an injective mapping such that the two vectors $\ba_\alpha=\d_\alpha\btheta$ are linearly independent at all points of $\bar{\omega}$, let $a^{\alpha\beta}$ denote the contravariant components of the metric tensor of $S=\btheta(\bar{\omega})$.  Let us consider the contravariant components of the scaled fourth order two-dimensional  tensor of the shell, $\a$  defined in (\ref{tensor_a_bidimensional}). Assume that $\lambda\geq0$ and $\mu>0$. Then there exists a constant $c_e>0$ independent of the variables and $\var$,  such that
 \begin{equation} \label{tensor_a_elip}
   \sum_{\alpha,\beta}|t_{\alpha\beta}|^2\leq c_e \a(\by)t_{\sigma\tau}t_{\alpha\beta},
 \end{equation}
 for all $\by\in\bar{\omega}$ and all $\bt=(t_{\alpha\beta})\in\mathbb{S}^2$.
 \end{theorem}

\begin{theorem}\label{teo_parte2}
Assume that  $V_0(\omega)=\{\bcero\}$. Then, $\bxi^0$ verifies the following variational formulation of a two-dimensional contact problem for elastic membrane shells: Find $\bxi^0:\omega \longrightarrow \mathbb{R}^3$ such that,
    \begin{align}\nonumber
    & \bxi^0\in K(\omega),\nonumber\\
   &\int_{\omega} \a\gst(\bxi^0)\gab(\beeta-\bxi^0)\sqrt{a}dy%
    \ge\int_{\omega}p^{i,0}(\eta_i-\xi_i^0)\sqrt{a}dy \ \forall \beeta=(\eta_i)\in K(\omega),\label{ec_paso4}
   \end{align}
   where,
   \begin{align}\label{p0}
   p^{i,0}:=\int_{-1}^{1}\fb^{i,0}dx_3+h_+^{i,1} \quad \textrm{and} \quad h_+^{i,1}=\bh^{i,1}(\cdot,+1).
   \end{align}
\end{theorem}

Before providing the proof we need some auxiliary results, which are given in the following lemmas:


\begin{lemma}\label{step1}
The negative order terms of the scaled linearized strains are null, i.e.,
$$
\eij^{-1}=0 \en \Omega.
$$
\end{lemma}

\begin{proof}
Since $\bu^0$ is independent of $x_3$ (see Theorem \ref{teo_parte1}), by using (\ref{eij_terminos_expansion_u}) and (\ref{eij_terminos_expansion}) we obtain that
\begin{equation}\label{eij_1}
\eij^{-1}=\eij^{-1}(\bu^0)=0 \en \Omega.
\end{equation}
\end{proof}


\begin{lemma}\label{step2}
The zeroth-order terms of the scaled linearized strains are identified. On one hand,
\begin{equation}\label{ea3}
\eatres^0=0 \en \Omega.
\end{equation}
On the other hand, we obtain that
\begin{align}\label{edtres_cero}
\edtres^0= - \frac{\lambda}{\lambda+2\mu}a^{\alpha \beta }\eab^0 \en \Omega.
\end{align}
\end{lemma}
\begin{proof}
Let $p=-1$ in (\ref{ecuacion_orden_fuerzas}). Grouping the terms multiplied by $\var^{-1}$ and using (\ref{eij_1}), we find that
\begin{align*}
&\intO A^{ijkl}(0)\ekl^0\eij^{-1}(\bv-\bu^0)\sqrt{a}dx\ge\intO f^{i,-1}(v_i-u_i^0) \sqrt{a}dx + \intG h^{i,0}(v_i-u_i^0) \sqrt{a} d\Gamma,
\end{align*}
for all $\bv\in K(\Omega)$. By considering a test function $\bv$ independent of $x_3$, and using similar arguments as in the proof of Theorem \ref{teo_parte1}, we obtain  that $f^{i,-1}$ and $h^{i,0}$ must be zero (recall (\ref{eij_1}), as well). Therefore, from the left-hand side of the last inequality we have,
\begin{align} \nonumber
&0\le\intO A^{ijkl}(0)\ekl^0\eij^{-1}(\bv)\sqrt{a}dx\nonumber\\
&\quad =\intO 4A^{\alpha 3 \sigma 3}(0)\eatres^0\estres^{-1}(\bv)\sqrt{a}dx%
+\intO\left( A^{\alpha\beta 33}(0)\eab^0 + A^{3333}(0)\edtres^0  \right)\edtres^{-1}(\bv) \sqrt{a}dx\nonumber\\
&\quad= \intO \left(2\mu a^{\alpha \sigma}\eatres^0\d_3v_\sigma + \left(\lambda a^{\alpha\beta}\eab^0+ (\lambda + 2\mu)\edtres^0 \right)\d_3v_3\right)\sqrt{a}dx,\label{ecuacion_integral1}
\end{align}
for all $\bv\in K(\Omega)$. If we take $\bv\in K(\Omega)$ such that $v_2=v_3=0$ and $v_1=v_3=0$ alternatively, we have the inequality
$$
\intO 2\mu a^{\alpha \sigma}\eatres^0\d_3v_\sigma \sqrt{a}dx\ge0.
$$
The equality is consequence of taking both $v_\sigma$ and its opposite (recall that $v_3=0$), and by using Theorem \ref{th_int_nula}, we conclude
\begin{align*}
e_{\alpha||3}^0=0 \en \Omega.
\end{align*}
Going back to (\ref{ecuacion_integral1}), we obtain
\begin{equation}\label{x1}
\intO  \left(\lambda a^{\alpha\beta}\eab^0+ (\lambda + 2\mu)\edtres^0 \right)\d_3v_3\sqrt{a}dx\ge0,
\end{equation}
for all $v_3\in H^1(\Omega)$ with $v_3=0 \en \Gamma_0$ and $v_3\ge0$ on $\Gamma_C$. In particular, the previous inequality is valid for the test functions in Theorem \ref{th_int_post}, and therefore we obtain
\begin{align} \label{ecuacion_casuistica}
\lambda a^{\alpha\beta} \eab^0 + (\lambda+2\mu) \edtres^0=0.
\end{align}

That is, we find (\ref{edtres_cero}).
\end{proof}

We are now ready to give the proof of Theorem \ref{teo_parte2}:

\begin{proof}[Theorem \ref{teo_parte2}]

Let $p=0$ in (\ref{ecuacion_orden_fuerzas}). Grouping the terms multiplied by $\var^{-1}$ and $\var^0$, and taking into account (\ref{tensorA_tildes}) and (\ref{eij_1}), we find that
\begin{align}
& \var^{-1}\intO A^{ijkl}(0)\ekl^0\eij^{-1}(\bv)\sqrt{a}dx+\intO A^{ijkl}(0) \ekl^0 (\eij^0(\bv)- \eij^0) \sqrt{a} dx\nonumber\\%
&\quad+\intO A^{ijkl}(0) \ekl^1\eij^{-1}(\bv) \sqrt{a} dx + \intO \tilde{A}^{ijkl,1}\ekl^0\eij^{-1}(\bv)dx\nonumber\\
& \quad \ge \intO f^{i,0} (v_i-u_i^0) \sqrt{a} dx + \intG h^{i,1}(v_i-u_i^0) \sqrt{a} d\Gamma,\label{ec_step3}
\end{align}
for all $\bv\in K(\Omega)$. Taking  $\bv\in K(\Omega)$ independent of the transversal variable $x_3$, it can be identified with a function $\beeta\in K(\omega)$, and we have by (\ref{eij_terminos_expansion}) that $\eij^{-1}(\bv)=0$. Moreover, since $\eatres^0=0$ by (\ref{ea3}), and since $\eab^0=\eab^0(\bu^0)=\eab^0(\bxi^0)$, we have
\begin{align}
&\intO A^{ijkl}(0)\ekl^0 (\eij^0(\beeta)-\eij^0) \sqrt{a} dx\nonumber\\
& \quad = \intO \left( \lambda a^{\alpha \beta} a^{\sigma \tau} + \mu (a^{\alpha \sigma}a^{\beta \tau} + a^{\alpha \tau}a^{\beta\sigma})  \right) \est^0 \eab^0(\beeta-\bxi^0) \sqrt{a} dx\nonumber\\%
&\qquad+ \intO \lambda a^{\alpha\beta}\edtres^0 \eab^0 (\beeta-\bxi^0)\sqrt{a} dx-\intO(\lambda+2\mu)\edtres^0\edtres^0\sqrt{a} dx\nonumber\\
&\quad \ge \intO f^{i,0} (\eta_i-\xi_i^0) \sqrt{a} dx + \intG h^{i,1}(\eta_i-\xi_i^0) \sqrt{a} d\Gamma .\label{ref2}
\end{align}
Now, by using the expression of $\edtres^0$ found in (\ref{edtres_cero}), we 
have that
\begin{align}
&\frac{1}{2}\intO \a \est^0(\bxi^0)\eab^0(\beeta-\bxi^0)\sqrt{a}dx \nonumber\\
& \quad \ge \intO f^{i,0}(\eta_i-\xi_i^0) \sqrt{a}dx + \intG h^{i,1} (\eta_i-\xi_i^0) \sqrt{a} d \Gamma, \ \forall\beeta\in K(\omega),\label{5.22b}
\end{align}
where  $\a$ denotes the contravariant components of the fourth order two-dimensional elasticity tensor defined in (\ref{tensor_a_bidimensional}).

Further, notice that if $\beeta=(\eta_i)\in H^1(\omega)\times H^1(\omega) \times L^2(\omega),$ then $\gab(\beeta)\in L^2(\omega)$. Hence, the equalities
\begin{equation} \label{et3}
\eab^0=\gab(\bxi^0), \quad  \eab^0(\beeta)=\gab(\beeta)\ \forall\ \beeta\in K(\omega),
\end{equation}
follow from the definitions (\ref{eij_terminos_expansion_u}), (\ref{eij_terminos_expansion}) and (\ref{def_gab}).
We end by combining (\ref{5.14b}), (\ref{5.22b}) and (\ref{et3}).
\end{proof}

\begin{remark}
The limit problem in Theorem \ref{teo_parte2} can be described as the scaled variational formulation of a two-dimensional unilateral contact problem  for an elastic membrane shell. More precisely, it can be described as an obstacle problem, since now the conditions are {\em in} the domain, $\omega$, while in contact problems as the original three-dimensional problem, the conditions are {\em on} the boundary of the domain, $\Omega$.
\end{remark}

\begin{remark}
Since by hipothesis, we have that $V_0(\omega)=\{\bcero\}$, it is easy to show that $\bxi^0$ is unique. The existence requires to pose the problem in a space wider than $V(\omega)$. This issue is discussed in the next section of this paper.
\end{remark}


Now, let
\begin{equation}\label{rab}
\rho_{\alpha\beta}(\beeta)= \d_{\alpha\beta}\eta_3 - \Gamma_{\alpha\beta}^\sigma \d_\sigma\eta_3 - b_\alpha^\sigma b_{\sigma\beta} \eta_3 + b_\alpha^\sigma (\d_\beta\eta_\sigma- \Gamma_{\beta\sigma}^\tau \eta_\tau)
+ b_\beta^\tau(\d_\alpha\eta_\tau-\Gamma_{\alpha\tau}^\sigma\eta_\sigma ) + b^\tau_{\beta|\alpha} \eta_\tau,
\end{equation}
denote the covariant components of the linearized change of curvature tensor associated with a displacement field $\beeta=\eta_i \ba^i$ of the surface $S$.

\begin{theorem}\label{teo_parte3}
Assume that that $\bu^1\in K(\Omega)$ and $K_F(\omega)\neq\{\bcero\}$. Then $\bxi^0$  is solution of the following variational formulation of a two-dimensional contact problem for elastic flexural shells: Find $\bxi^0:\omega \longrightarrow\mathbb{R}^3$ such that,
    \begin{align*}
    & \bxi^0\in K_F(\omega),\\
   &\frac{1}{3}\int_{\omega} \a\rst(\bxi^0)\rab(\beeta-\bxi^0)\sqrt{a}dy %
   \ge\int_{\omega}p^{i,2}(\eta_i-\xi_i^0)\sqrt{a}dy \ \forall \beeta=(\eta_i)\in K_F(\omega),
   \end{align*}

   where,
   \begin{align} \label{p2}
 p^{i,2}:=\int_{-1}^{1}\fb^{i,2}dx_3+h_+^{i,3}, \quad \textrm{and} \quad h_{+}^{i,3}=\bh^{i,3}(\cdot,+1).
   \end{align}
\end{theorem}

Before providing the proof of this theorem we need some auxiliary results, which are given in the next lemma.


\begin{lemma}\label{step5}
Assume that $K_0(\omega)\neq\{\bcero\}$. We find that
$$
\eij^0=0 \en \Omega,\ \textrm{and}\ \ \bxi^0\in K_F(\omega).
$$
Moreover, assume that $\bu^1\in K(\Omega)$. Then, there exists a function $\bxi^1=(\xi_i^1)\in K(\omega)$, such that
$$
u_\alpha^1=\xi_\alpha^1 - x_3( \d_\alpha\xi_3^0 + 2 b_\alpha^\sigma \xi_\sigma^0),\qquad
u_3^1=\xi_3^1.
$$
Also, the following first-order terms of the scaled linearized strains  are identified:
\begin{equation}\label{ultima}
\eatres^1=0,\qquad
\edtres^1= - \frac{\lambda}{\lambda + 2\mu} a^{\alpha \beta} \eab^1\ \en \Omega.
\end{equation}
Moreover,
\begin{equation}\label{ref1}
\eab^1=\gab(\bxi^1)- x_3\rab(\bxi^0).
\end{equation}
\end{lemma}

\begin{proof}
For (\ref{ec_paso4}) $\beeta\in K_0(\omega)\setminus\{\bcero\}$, we find that
\begin{align*}
\into p^{i,0}(\eta_i-\xi_i^0) \sqrt{a}dy = \intO f^{i,0}(\eta_i-\xi_i^0) \sqrt{a}dx + \intG h^{i,1} (\eta_i-\xi_i^0) \sqrt{a} d \Gamma\le0.
\end{align*}
Hence, in order to avoid compatibility conditions between the applied forces we must take $f^{i,0}=0$ and $h^{i,1}=0$. Then, taking $\beeta=2\bxi^0$ and $\beeta=\bzero$ in the inequality (\ref{ec_paso4}) leads to
\begin{align*} \nonumber
\into \a \gst(\bxi^0)\gab(\bxi^0)\sqrt{a}dy=0,
\end{align*}
which implies that $\gab(\bxi^0)=0$, that is, $\bxi^0\in V_0(\omega)$. Therefore, by (\ref{et3}), we find that $\eab^0=\gab(\bxi^0)=0$. Moreover, by (\ref{eij_terminos_expansion_u}) and (\ref{edtres_cero}) we have that
\begin{equation*}
\d_3 u_3^1=\edtres^0=0 \en \Omega.
\end{equation*}
By the definition of $\eatres^0$ in (\ref{eij_terminos_expansion_u}) and using (\ref{ea3}), we have that
\begin{align*}
\eatres^0=\frac{1}{2}\left( \d_\alpha \xi^0_3+ \d_3 u_\alpha^1 \right)+ b_\alpha^\sigma \xi^0_\sigma =0,
\end{align*}
hence,
\begin{align*}
\d_3 u_\alpha^1=- \left(\d_\alpha \xi_3^0 + 2b_\alpha^\sigma \xi^0_\sigma\right)  \en \Omega.
\end{align*}
Since, in particular, we are assuming that $\bu^1\in V(\Omega)$ and since $\bxi^0$ is independent of $x_3$ by Theorem \ref{teo_parte1}, there exists a field $\bxi^1\in V(\omega)$ such that
$$
u_\alpha^1= \xi^1_\alpha - x_3 \left(\d_\alpha\xi_3^0 + 2 b_\alpha^\sigma \xi_\sigma^0 \right),\quad
u_3^1=\xi_3^1,
$$
in $\Omega$. Notice that the first equality above gives that $\xi^0_3\in H^2(\Omega)$. Further, since $\xi_\alpha^0=0$ on $\gamma_0$, then $\d_\nu\xi_3^0=0$ on $\gamma_0$. Therefore, we have $\bxi^0\in V_F(\omega)$. Moreover, since $\bu^1\in K(\Omega)$ and $\partial_3 u_3^1=0$, this implies that $u_3^1=\xi_3^1\ge0$. Therefore, $\bxi^1\in K(\omega)$. Since $\eij^0=0$, coming back to the terms multiplied by $\var^0$ (see (\ref{ec_step3})), we have
\begin{align*}
\intO A^{ijkl}(0) \ekl^1\eij^{-1}(\bv)  \sqrt{a} dx\ge0,
\end{align*}
for all $\bv\in V(\Omega)$. Notice that this inequality is analogous to (\ref{ecuacion_integral1}) but now involving the terms $\eij^1$ instead of the terms $\eij^0$. Therefore, using similar arguments as in Lemma \ref{step2}, we deduce the expressions in (\ref{ultima}).
Now by the the definitions in (\ref{eij_terminos_expansion_u}) in terms of $\xi^0_i$ and $\xi^1_i$ and replacing $\d_\beta b_\alpha^\sigma$ terms from (\ref{b_barra}), after some computations  we have that
\begin{align}
&\eab^1= \frac{1}{2} \left( \d_\beta \xi_\alpha^1 + \d_\alpha \xi_\beta^1  \right) - \Gamma_{\alpha\beta}^\sigma \xi_\sigma^1 - b_{\alpha\beta} \xi_3^1%
-x_3\Big( \d_{\alpha\beta}\xi_3^0 - \Gamma_{\alpha\beta}^\sigma\d_\sigma\xi_3^0\nonumber\\
&\quad - b_\alpha^\sigma b_{\sigma\beta}\xi_3^0 %
+ b_\alpha^\sigma\left(\d_\beta\xi_\sigma^0 - \Gamma_{\beta\sigma}^\tau\xi_\tau^0    \right)   + b_\beta^\tau\left( \d_\alpha\xi_\tau^0 - \Gamma_{\alpha\tau}^\sigma \xi_\sigma^0 \right)+ b^\tau_{\beta|\alpha}\xi_\tau^0\Big).\label{ec_paso5}
\end{align}
Note that if $\beeta=(\eta_i)\in H^1(\omega)\times H^1(\omega) \times L^2(\omega),$ then (see (\ref{rab})) it is verified that
$\rho_{\alpha\beta} (\beeta) \in L^2(\Omega)$. Hence, by (\ref{def_gab}) for $\beeta=\bxi^1$ and (\ref{rab}) for $\beeta=\bxi^0$, it follows from (\ref{ec_paso5}) the equality (\ref{ref1}).
\end{proof}


We are now ready to give the proof of Theorem \ref{teo_parte3}.

\begin{proof}[Theorem \ref{teo_parte3}]
Let $p=1$ in (\ref{ecuacion_orden_fuerzas}).
Note that, taking into account the results in the previous lemmas and theorems and, particularly, since $\eij^{-1}=\eij^0=0$, we can go on grouping and canceling the following order terms from the original inequality. Therefore, by grouping the terms multiplied by $\var$, we have
\begin{align}\nonumber
&\intO A^{ijkl}(0)\left(\ekl^1\eij^0(\bv-\bxi^0) + \ekl^2\eij^{-1}(\bv-\bxi^0) \right) \sqrt{a} dx\\
&\ + \intO \tilde{A}^{ijkl,1}\ekl^1\eij^{-1}(\bv-\bxi^0)dx%
\ge\intO f^{i,1}(v_i-\xi_i^0) \sqrt{a}dx + \intG h^{i,2} (v_i-\xi_i^0) \sqrt{a}d\Gamma, \label{ref3}
\end{align}
for all $\bv\in K(\Omega)$. Taking $\bv=\beeta\in K(\omega)$ (which implies that $\bv$ is independent of $x_3$), by (\ref{eij_terminos_expansion}) we obtain
\begin{align*}
&\intO A^{ijkl}(0)\ekl^1\eij^0(\beeta-\bxi^0) \sqrt{a} dx
\ge\intO f^{i,1}(\eta_i-\xi_i^0) \sqrt{a}dx + \intG h^{i,2} (\eta_i-\xi_i^0) \sqrt{a}d\Gamma,
\end{align*}
for all $\beeta\in K(\omega)$. Since $\eatres^1=0$ (see Lemma \ref{step5}), we obtain
\begin{align*}
&\intO A^{ijkl}(0)\ekl^1\eij^0(\beeta-\bxi^0) \sqrt{a} dx\\
&\quad =\intO \left( \lambda a^{\alpha \beta} a^{\sigma \tau} + \mu (a^{\alpha \sigma}a^{\beta \tau} + a^{\alpha \tau}a^{\beta\sigma})  \right) \est^1 \eab^0(\beeta-\bxi^0) \sqrt{a} dx\\
&\qquad+ \intO \lambda a^{\alpha\beta}\edtres^1 \eab^0 (\beeta-\bxi^0)\sqrt{a} dx\\
&\quad \ge \intO f^{i,1} (\eta_i-\xi_i^0)\sqrt{a} dx + \intG h^{i,2}(\eta_i-\xi_i^0) \sqrt{a} d\Gamma ,
\end{align*}
for all $\beeta\in K(\omega)$, which is analogous to the expression obtained in (\ref{ref2}). Therefore, following the same arguments made there, taking into account Lemma \ref{step5} and using (\ref{ref1}),  we find that
\begin{align*}
&\into \a \gst(\bxi^1)\gab(\beeta-\bxi^0)\sqrt{a}dy \ge \intO f^{i,1}(\eta_i-\xi_i^0)\sqrt{a} dx\nonumber\\
&\qquad + \intG h^{i,2}(\eta_i-\xi_i^0) \sqrt{a} d\Gamma ,
\end{align*}
for all $\beeta\in K(\omega)$. Note that, in particular, we are assuming that $K_0(\omega)\neq\{\bcero\}$. Therefore, for $\beeta\in K_0(\omega)\setminus\{\bcero\}$ and since $\gab(\bxi^0)=0$, we have that
\begin{align*}
\intO f^{i,1} (\eta_i-\xi_i^0) \sqrt{a} dx + \intG h^{i,2}(\eta_i-\xi_i^0) \sqrt{a} d\Gamma \le0,
\end{align*}
hence, in order to avoid compatibility conditions between the applied forces we must take $f^{i,1}=0$ and $h^{i,2}=0$.
On one hand, coming back to inequality (\ref{ref3}), we find that
\begin{align*}\nonumber
&\intO A^{ijkl}(0)\left(\ekl^1\eij^0(\bv-\bxi^0) + \ekl^2\eij^{-1}(\bv-\bxi^0) \right) \sqrt{a} dx\\
&\qquad+ \intO \tilde{A}^{ijkl,1}\ekl^1\eij^{-1}(\bv-\bxi^0)dx\ge0,
\end{align*}
for all $\bv\in K(\Omega)$. Given $\beeta\in K_F(\omega)$, we define $\bv(\beeta)=(v_i(\beeta))\in K(\Omega)$ as
\begin{equation*}
v_\alpha(\beeta):= x_3 \left( 2b_\alpha^\sigma \eta_\sigma + \d_\alpha\eta_3 \right),\quad
v_3(\beeta):=0,
\end{equation*}
and take $\bv=\bv(\beeta)$ in the previous inequality, thus leading to (see (\ref{eij_terminos_expansion}))
\begin{align}\nonumber
&\intO A^{ijkl}(0)\ekl^1\eij^0(\bv(\beeta)) \sqrt{a}dx + 4\intO  A^{\alpha 3\sigma 3}(0)\estres^2\left( b_\alpha^\tau\eta_\tau + \frac{1}{2} \d_\alpha\eta_3   \right) \sqrt{a} dx
\\\nonumber
& \qquad + 4\intO \tilde{A}^{\alpha 3 \sigma 3,1}\estres^1\left( b_\alpha^\tau\eta_\tau + \frac{1}{2} \d_\alpha\eta_3   \right)dx\ge0,
\end{align}
for all $\beeta\in K_F(\omega)$. Now, if we repeat the previous process by using $-\bv(\beeta)$ which still belongs to $K(\Omega)$, since $v_3(\beeta)=0$, we obtain the opposite inequality, and therefore, the equality holds:
\begin{align}
&\intO A^{ijkl}(0)\ekl^1\eij^0(\bv(\beeta)) \sqrt{a}dx + 4\intO  A^{\alpha 3\sigma 3}(0)\estres^2\left( b_\alpha^\tau\eta_\tau + \frac{1}{2} \d_\alpha\eta_3   \right) \sqrt{a} dx\nonumber\\
& \qquad + 4\intO \tilde{A}^{\alpha 3 \sigma 3,1}\estres^1\left( b_\alpha^\tau\eta_\tau + \frac{1}{2} \d_\alpha\eta_3   \right)dx=0.\label{x3}
\end{align}
On the other hand, let $p=2$ in  (\ref{ecuacion_orden_fuerzas}). Grouping the terms multiplied by $\var^2$ and using the results in lemmas \ref{step1} and \ref{step5}, we find that
\begin{align*}
&\var^{-1}\intO A^{ijkl}(0) \left(\ekl^1\eij^0(\bv)+\ekl^2\eij^{-1}(\bv)\right)\sqrt{a}dx%
 + \var^{-1}\intO \tilde{A}^{ijkl,1}\ekl^1\eij^{-1}(\bv)dx\\
&\ +\intO A^{ijkl}(0) \left(\ekl^1(\eij^1(\bv)-\eij^1) + \ekl^2\eij^0(\bv-\bxi^0) + \ekl^{3}\eij^{-1}(\bv-\bxi^0)\right)\sqrt{a}dx\\
&\ + \intO \tilde{A}^{ijkl,1}\left(\ekl^1\eij^0(\bv-\bxi^0) + \ekl^2\eij^{-1}(\bv-\bxi^0)   \right) dx +
\intO \tilde{A}^{ijkl,2} \ekl^1\eij^{-1}(\bv-\bxi^0) dx\\
&  \ge \intO f^{i,2} (v_i-\xi_i^0) \sqrt{a} dx + \intG h^{i,3}(v_i-\xi_i^0) \sqrt{a} d\Gamma ,
\end{align*}
for all $\bv\in K(\Omega)$. Consider now any $\bv$ which can be identified with a function $\beeta\in K_F(\omega)$.  By using (\ref{eij_terminos_expansion}) we deduce
$$
\eij^{-1}(\beeta)=0,\ \eab^0(\beeta)=\gab(\beeta)=0,\ \edtres^0(\beeta)=0.
$$
Hence by taking into account again the results in lemmas \ref{step1} and \ref{step5} (particularly, recall that $\estres^1=0$), we have
\begin{align*}
&\intO A^{ijkl}(0) \ekl^1(\eij^1(\beeta)-\eij^1)\sqrt{a}dx \\
&\quad+ 4\intO A^{\alpha 3 \sigma 3}(0) \estres^2 \left(  b_\alpha^\tau\eta_\tau + \frac{1}{2} \d_\alpha\eta_3   \right) \sqrt{a}dx\ge\into p^{i,2} (\eta_i-\xi_i^0) \sqrt{a} dy,
\end{align*}
 for all $\beeta\in K_F( \omega)$, where $p^{i,2}$ is defined  in (\ref{p2}). By subtracting (\ref{x3}), we obtain
\begin{align}\label{ref5}
&\intO A^{ijkl}(0) \ekl^1 \left(\eij^1(\beeta) - \eij^0(\bv(\beeta))-\eij^1\right) \sqrt{a}dx%
\ge \into p^{i,2} (\eta_i-\xi_i^0) \sqrt{a} dy ,
\end{align}
for all $\beeta\in K_F( \omega)$. Now,  by having in mind Lemma \ref{step5} and (\ref{eij_terminos_expansion})  we have that
\begin{align*}
&A^{ijkl}(0)\ekl^1 \left( \eij^1(\beeta) - \eij^0(\bv(\beeta))-\eij^1\right)\\%
&\ = A^{\alpha\beta\sigma\tau}(0)\est^1\left(\eab^1 (\beeta)- \eab^0(\bv(\beeta))-\eab^1 \right)\\
&\quad+A^{\alpha\beta 33}(0)\edtres^1\left(\eab^1 (\beeta)- \eab^0(\bv(\beeta))-\eab^1\right)%
-A^{3333}(0)\edtres^1\edtres^1.
\end{align*}
Note that the contribution of the last term above to the left-hand side of (\ref{ref5}) can be removed while keeping the inequality.
Furthermore, by (\ref{eij_terminos_expansion}) we also find that
\begin{align*}
&\eab^1(\beeta)-\eab^0(\bv(\beeta))%
= x_3 \left( b^\sigma_{\beta|\alpha} \eta_\sigma + b_\alpha^\sigma b_{\sigma\beta}\eta_3  \right)  \\
&\quad-x_3 \left(\d_\alpha(b_\beta^\tau \eta_\tau)+\d_\beta(b_\alpha^\sigma\eta_\sigma) + \d_{\alpha\beta}\eta_3 - \Gamma_{\alpha\beta}^\sigma \d_\sigma \eta_3 -2\Gamma_{\alpha\beta}^\sigma b_\sigma^\tau \eta_\tau   \right),
\end{align*}
and making some calculations we conclude that
\begin{align*}
\eab^1(\beeta) - \eab^0(\bv(\beeta)) = - x_3 \rab(\beeta), \ \forall \beeta\in V_F(\omega).
\end{align*}
Therefore, by the arguments and calculations above, and using (\ref{ref1}) and(\ref{ultima}), the left-hand side of the inequality (\ref{ref5}) leads to
\begin{align*}
&\intO A^{ijkl}(0) \ekl^1 \left(\eij^1(\beeta) - \eij^0(\bv(\beeta))-\eij^1\right) \sqrt{a}dx\nonumber\\
& \quad  \le\intO \left( \lambda a^{\alpha \beta} a^{\sigma \tau} + \mu (a^{\alpha \sigma}a^{\beta \tau} + a^{\alpha \tau}a^{\beta\sigma})  \right)\est^1\left( -x_3 \rab(\beeta-\bxi^0)-\gab(\bxi^1)\right) \sqrt{a} dx \nonumber\\
&\qquad + \intO \lambda a^{\alpha\beta}\edtres^1 \left(-x_3 \rab(\beeta-\bxi^0)-\gab(\bxi^1) \right)\sqrt{a} dx \nonumber\\
&\quad=\intO \frac{x_3^2}{2}\a \rst(\bxi^0)\rab(\beeta-\bxi^0)\sqrt{a}dx%
-\intO \frac{1}{2}\a \gab(\bxi^1)\gst(\bxi^1)\sqrt{a}dx\\
&\quad\le \frac{1}{3}\into \a \rst(\bxi^0)\rab(\beeta-\bxi^0)\sqrt{a}dy,
\end{align*}
for all $\beeta\in K_F(\omega)$.  Hence, going back with this result to (\ref{ref5}), we obtain the limit problem formulated in the theorem.
\end{proof}

\begin{remark}
The problem formulated in Theorem \ref{teo_parte3} can be described as the scaled version of the variational formulation of a two-dimensional unilateral contact problem (or, more precisely, obstacle problem) for an elastic flexural shell. The existence and uniqueness of solution is discussed in the next section.
\end{remark}
\begin{remark}
Note that, unlike the non contact case, studied in \cite{Ciarlet4b}, we do not find that $\bxi^1 \in K_0(\omega)$, not even under the assumption that $K_0(\omega)\neq\{\bcero\}$. 
\end{remark}

\section{Existence and uniqueness of the solution of the two-dimensional problems} \label{Existencia}

{ In what follows, we show the existence and uniqueness of solution for some of the two-dimensional limit problems found in the previous section, namely the elliptic membrane shell and the flexural shell. The study of the remaining cases of membranes are more technically involved and will be presented in a forthcoming paper.

We shall present both limit problems in a de-scaled form. The scalings in Section \ref{seccion_dominio_ind} suggest the de-scalings $\xi_i^\var(\by)=\xi_i^0(\by)$ for all $\by\in\bar{\omega}$.
}

\subsection{Elastic elliptic membrane shell contact problem}


{ When the middle surface $S$ of a membrane is elliptic, and $\gamma=\gamma_0$, the right space to find a solution is $V_M(\omega):=H^1_0(\omega)\times H^1_0(\omega)\times L^2(\omega)$. Accordingly, we define $K_M(\omega):=\{\beeta\in V_M(\omega),\ \eta_3\ge0\}$ as the subset of admissible displacements.} For this type of membranes it is verified the two-dimensional Korn's type inequality (see, for example, Theorem 2.7-3, \cite{Ciarlet4b}): there exists a constant $c_M=c_M(\omega,\btheta)$ such that
 \begin{align} \label{Korn_elipticas}
 \left( \sum_\alpha||\eta_\alpha||^2_{1,\omega} + ||\eta_3||_{0,\omega}^2       \right)^{1/2} \leq c_M \left( \sum_{\alpha,\beta} ||\gab(\beeta)||_{0,\omega}^2  \right)^{1/2} \ \forall \beeta\in V_M(\omega).
 \end{align}
{ Note that this implies that $V_0(\omega)=\{\bcero\}$, as required in Theorem \ref{teo_parte2}}.
Therefore, we can enunciate the de-scaled variational formulation of the contact problem for an elastic elliptic membrane shell as follows:

\begin{problem}\label{problema_ab_eps}
Find $\bxi^\var:\omega \longrightarrow \mathbb{R}^3$ such that,
   \begin{align*}
   &\bxi^\var\in K_M(\omega),\
   \var\int_{\omega} \aeps\gst(\bxi^\var)\gab(\beeta-\bxi^\var)\sqrt{a}dy\\
&\qquad   \ge\int_{\omega}p^{i,\var}(\eta_i-\xi_i^\var)\sqrt{a}dy \ \forall \beeta=(\eta_i)\in K_M(\omega),
   \end{align*}
   where,
   \begin{align*}
   &\gab(\beeta):= \frac{1}{2}(\d_\alpha\eta_\beta + \d_\beta\eta_\alpha) - \Gamma_{\alpha\beta}^\sigma\eta_\sigma -b_{\alpha\beta}\eta_3,   \\
   & p^{i,\var}:=\int_{-\var}^{\var}\fb^{i,\var}dx_3^\var +h_+^{i,\var}, \ \textrm{and} \ h_{+}^{i,\var}=\bh^{i,\var}(\cdot,\var),
   \end{align*}
  and where the contravariant components of the fourth order two-dimensional tensor $\aeps,$ are defined as rescaled versions of  (\ref{tensor_a_bidimensional}).
\end{problem}

\begin{theorem} \label{Th_exist_unic_bid_cero}
Let $\omega$  be a domain in $\mathbb{R}^2$, let $\btheta\in\mathcal{C}^2(\bar{\omega};\mathbb{R}^3)$ be an injective mapping such that the two vectors $\ba_\alpha=\d_\alpha\btheta$ are linearly independent at all points of $\bar{\omega}$. Let $\fb^{i,\var}\in L^2(\Omega^\var) $, $\bh^{i,\var}\in L^2(\Gamma_+^\var)$. Then the Problem \ref{problema_ab_eps}, has a unique solution  $\bxi^\var\in K_M(\omega)$.
\end{theorem}
\begin{proof}
 Let us consider the bilinear symmetric form  $a^\var: V_M(\omega)\times V_M(\omega)\longrightarrow \mathbb{R}$ defined by,
\begin{align*}
a^\var(\bxi^\var,\beeta)&: = \var\int_{\omega} \aeps\gst(\bxi^\var)\gab(\beeta)\sqrt{a}dy,
\end{align*}
 for all $\bxi^\var,\beeta\in V_M(\omega) $ and for each $\var>0$. By the ellipticity of the two-dimensional elasticity tensor, established in (\ref{tensor_a_elip}), and by using a Korn's type inequality (see (\ref{Korn_elipticas})), we find out that $a^\var$ is elliptic and continuous in $V_M(\Omega)$. Also notice that $p^{i,\var}\in L^2(\omega)$ and $K_M(\omega)$ is a non-empty, closed and convex subset of $V_M(\omega)$. Therefore, the Problem \ref{problema_ab_eps}, which can be formulated as the elliptic variational inequality
 $$
 \bxi^\var\in K_M(\omega),\quad a^\var(\bxi^\var,\beeta-\bxi^\var)%
 \ge \int_{\omega}p^{i,\var}(\eta_i-\xi_i^\var)\sqrt{a}dy \ \forall \beeta=(\eta_i)\in K_M(\omega),
 $$
 has a unique solution (see, for example, {\cite{Capelo,Jarusek,HS}}).
\end{proof}

\subsection{Elastic flexural shell contact problem}

Let us consider now that the set $K_F(\omega)$ contains non-zero functions, as required in Theorem \ref{teo_parte3}. Therefore, we can enunciate the de-scaled variational formulation of the contact problem for an elastic flexural shell:

\begin{problem}\label{problema_flexural_eps}
 Find $\bxi^\var:\omega \longrightarrow \mathbb{R}^3$ such that,
    \begin{align*}
    & \bxi^\var\in K_F(\omega),\
   \frac{\var^3}{3}\int_{\omega} \aeps\rst(\bxi^\var)\rab(\beeta-\bxi^\var)\sqrt{a}dy\\
&\qquad   \ge\int_{\omega}p^{i,\var}(\eta_i-\xi_i^\var)\sqrt{a}dy \ \forall \beeta=(\eta_i)\in K_F(\omega),
   \end{align*}
   where,
   \begin{align*}
   &\rho_{\alpha\beta}(\beeta):= \d_{\alpha\beta}\eta_3 - \Gamma_{\alpha\beta}^\sigma \d_\sigma\eta_3 - b_\alpha^\sigma b_{\sigma\beta} \eta_3 + b_\alpha^\sigma (\d_\beta\eta_\sigma- \Gamma_{\beta\sigma}^\tau \eta_\tau)\\
   &\qquad\qquad+ b_\beta^\tau(\d_\alpha\eta_\tau-\Gamma_{\alpha\tau}^\sigma\eta_\sigma ) + b^\tau_{\beta|\alpha} \eta_\tau,   \\
   & p^{i,\var}:=\int_{-\var}^{\var}\fb^{i,\var}dx_3^\var+h_+^{i,\var},\ \textrm{and} \ h_{+}^{i,\var}=\bh^{i,\var}(\cdot,\var),
   \end{align*}
 and where the contravariant components of the fourth order two-dimensional tensor $\aeps$ are defined as rescaled versions of  (\ref{tensor_a_bidimensional}).
\end{problem}

If $\btheta\in\mathcal{C}^3(\bar{\omega};\mathbb{R}^3),$ it is verified the following Korn's type inequality (see, for example, Theorem 2.6-4, \cite{Ciarlet4b}): there exists a constant $c=c(\omega, \gamma_0, \btheta)$ such that
 \begin{align} \label{Korn_flexural}
 \left( \sum_\alpha||\eta_\alpha||^2_{1,\omega} + ||\eta_3||_{2,\omega}^2       \right)^{1/2} \leq c \left( \sum_{\alpha,\beta} ||\rab(\beeta)||_{0,\omega}^2  \right)^{1/2} \ \forall \beeta\in V_F(\omega).
 \end{align}
\begin{theorem} \label{Th_exist_unic_bid_dos}
Let $\omega$  be a domain in $\mathbb{R}^2$, let $\btheta\in\mathcal{C}^3(\bar{\omega};\mathbb{R}^3)$ be an injective mapping such that the two vectors $\ba_\alpha=\d_\alpha\btheta$ are linearly independent at all points of $\bar{\omega}$. Let $\fb^{i,\var}\in L^2(\Omega^\var)$, $\bh^{i,\var}\in L^2(\Gamma_+^\var)$. Then the Problem \ref{problema_flexural_eps}, has a unique solution  $\bxi^\var\in K_F(\omega)$.
\end{theorem}
\begin{proof}
 Let us consider the symmetric bilinear form $a^\var: V_F(\omega)\times V_F(\omega)\longrightarrow \mathbb{R}$  defined by,
\begin{align*}
a_F^\var(\bxi^\var,\beeta)&: = \frac{\var^3}{3}\int_{\omega} \aeps\rst(\bxi^\var)\rab(\beeta)\sqrt{a}dy,
\end{align*}
for all $\bxi^\var,\beeta\in V_F(\omega) $ and for each $\var>0$. By the ellipticity of the two-dimensional elasticity tensor established in (\ref{tensor_a_elip}), and by using a Korn's type inequality (see (\ref{Korn_flexural}), we find out that $a_F^\var$ is elliptic and continuous in $V_F(\Omega)$. Also notice that $p^{i,\var}\in L^2(\omega)$ and $K_F(\omega)$ is a non-empty, closed and convex subset of $V_F(\omega)$. Therefore, the Problem \ref{problema_flexural_eps}, which can be formulated as the elliptic variational inequality
 $$
 \bxi^\var\in K_F(\omega),\quad a_F^\var(\bxi^\var,\beeta-\bxi^\var)%
 \ge \int_{\omega}p^{i,\var}(\eta_i-\xi_i^\var)\sqrt{a}dy \ \forall \beeta=(\eta_i)\in K_F(\omega),
 $$
 has a unique solution.
\end{proof}

\section{Conclusions} \label{conclusiones}

We have identified two-dimensional variational formulations of obstacle problems for elastic membranes and flexural shells as the approximations of the three-dimensional variational formulation of the scaled, unilateral, frictionless, contact problem of an elastic shell.

To this end we used curvilinear coordinates and the asymptotic expansion method. We have provided an analysis of the existence and uniqueness of solution for these problems but, in the case of the membranes, the proof is limited to the case when it is elliptic. The proof of the existence and uniqueness of solution for the other cases of membranes is an issue for a forthcoming paper.

The asymptotic approaches found need to be fully justified with convergence theorems. Guided by the insight obtained from the formal analysis developed in this paper, these results will be presented in forthcoming papers.


\bibliographystyle{spmpsci}      
\bibliography{biblio_CS}   

\begin{thebibliography}{10}
\providecommand{\url}[1]{{#1}}
\providecommand{\urlprefix}{URL }
\expandafter\ifx\csname urlstyle\endcsname\relax
  \providecommand{\doi}[1]{DOI~\discretionary{}{}{}#1}\else
  \providecommand{\doi}{DOI~\discretionary{}{}{}\begingroup
  \urlstyle{rm}\Url}\fi

\bibitem{Capelo}
Baiocchi, C., Capelo, A.: Variational and quasivariational inequalities.
\newblock A Wiley-Interscience Publication. John Wiley \& Sons, Inc., New York
  (1984).
\newblock Applications to free boundary problems, Translated from the Italian
  by Lakshmi Jayakar

\bibitem{BM}
Berm{\'u}dez, A., Moreno, C.: Duality methods for solving variational
  inequalities.
\newblock Comput. Math. Appl. \textbf{7}(1), 43--58 (1981).
\newblock \doi{10.1016/0898-1221(81)90006-7}.
\newblock \urlprefix\url{http://dx.doi.org/10.1016/0898-1221(81)90006-7}

\bibitem{BV}
Berm\'udez, A., Via{\~n}o, J.M.: Une justification des \'equations de la
  thermo\'elasticit\'e de poutres \`a section variable par des m\'ethodes
  asymptotiques.
\newblock Math. Model. Numer. Anal. \textbf{18}(4), 347--376 (1984)

\bibitem{Ciarlet2}
Ciarlet, P.G.: Mathematical elasticity. {V}ol. {I}: Three-dimensional
  elasticity, \emph{Studies in Mathematics and its Applications}, vol.~20.
\newblock North-Holland Publishing Co., Amsterdam (1988)

\bibitem{Ciarlet3}
Ciarlet, P.G.: Mathematical elasticity. {V}ol. {II}: Theory of plates,
  \emph{Studies in Mathematics and its Applications}, vol.~27.
\newblock North-Holland Publishing Co., Amsterdam (1997)

\bibitem{Ciarlet4b}
Ciarlet, P.G.: Mathematical elasticity. {V}ol. {III}: Theory of shells,
  \emph{Studies in Mathematics and its Applications}, vol.~29.
\newblock North-Holland Publishing Co., Amsterdam (2000)

\bibitem{CD}
Ciarlet, P.G., Destuynder, P.: A justification of the two-dimensional linear
  plate model.
\newblock J. M\'ecanique \textbf{18}(2), 315--344 (1979)

\bibitem{CiarletLods2}
Ciarlet, P.G., Lods, V.: Asymptotic analysis of linearly elastic shells.
  justification of membrane shell equations.
\newblock Arch. Rational Mech. Anal. \textbf{136}, 119--161 (1996)

\bibitem{CiarletLods}
Ciarlet, P.G., Lods, V.: On the ellipticity of linear membrane shell equations.
\newblock J. Math. Pures Appl. \textbf{75}, 107--124 (1996)

\bibitem{CGLRT2}
Cimeti{\`e}re, A., Geymonat, G., Le~Dret, H., Raoult, A., Tutek, Z.: Asymptotic
  theory and analysis for displacements and stress distribution in nonlinear
  elastic straight slender rods.
\newblock J. Elasticity \textbf{19}(2), 111--161 (1988).
\newblock \doi{10.1007/BF00040890}.
\newblock \urlprefix\url{http://dx.doi.org/10.1007/BF00040890}

\bibitem{Destuynder}
Destuynder, P.: Sur une justification des mod\`eles de plaques et de coques par
  les m\'ethodes asymptotiques.
\newblock Ph.D. thesis, Univ. P. et M. Curie, Paris (1980)

\bibitem{DL}
Duvaut, G., Lions, J.L.: Inequalities in Mechanics and Physics.
\newblock Springer Berlin (1976)

\bibitem{Jarusek}
Eck, C., Jaru{\v{s}}ek, J., Krbec, M.: Unilateral contact problems, \emph{Pure
  and Applied Mathematics (Boca Raton)}, vol. 270.
\newblock Chapman \& Hall/CRC, Boca Raton, FL (2005).
\newblock \doi{10.1201/9781420027365}.
\newblock \urlprefix\url{http://dx.doi.org/10.1201/9781420027365}.
\newblock Variational methods and existence theorems

\bibitem{GLT}
Glowinski, R., Lions, J.L., Tr{\'e}moli{\`e}res, R.: Numerical analysis of
  variational inequalities, \emph{Studies in Mathematics and its Applications},
  vol.~8.
\newblock North-Holland Publishing Co., Amsterdam-New York (1981).
\newblock Translated from the French

\bibitem{HS}
Han, W., Sofonea, M.: Quasistatic Contact Problems in Viscoelasticity and
  Viscoplasticity.
\newblock AMS/IP Studies in Advanced Mathematics. American Mathematical Society
  / International Press, Providence-Somerville (2002)

\bibitem{HHNL}
Hlav\'{a}\v{c}ek, I., Haslinger, J., Nec\v{a}s, J., Lov\'{\i}\v{s}ek, J.:
  Solution of Variational Inequalities in Mechanics.
\newblock Applied Mathematical Sciences. Springer-Verlag, New York (1988)

\bibitem{iv}
Irago, H., Via{\~n}o, J.M.: Error estimation in the {B}ernoulli-{N}avier model
  for elastic rods.
\newblock Asymptot. Anal. \textbf{21}(1), 71--87 (1999)

\bibitem{KO}
Kikuchi, N., Oden, J.T.: Contact Problems in Elasticity: A Study of Variational
  Inequalities and Finite Element Methods, \emph{SIAM Studies in Applied
  Mathematics}, vol.~8.
\newblock SIAM, Philadelphia (1988)

\bibitem{Kunisch}
{Kunisch, K.}, {Stadler, G.}: Generalized newton methods for the 2d-signorini
  contact problem with friction in function space.
\newblock ESAIM: M2AN \textbf{39}(4), 827--854 (2005).
\newblock \doi{10.1051/m2an:2005036}.
\newblock \urlprefix\url{http://dx.doi.org/10.1051/m2an:2005036}

\bibitem{LM1}
L{\'e}ger, A., Miara, B.: Mathematical justification of the obstacle problem in
  the case of a shallow shell.
\newblock J. Elasticity \textbf{90}(3), 241--257 (2008).
\newblock \doi{10.1007/s10659-007-9141-1}.
\newblock \urlprefix\url{http://dx.doi.org/10.1007/s10659-007-9141-1}

\bibitem{LM2}
L{\'e}ger, A., Miara, B.: Erratum to: {M}athematical justification of the
  obstacle problem in the case of a shallow shell [mr2387957].
\newblock J. Elasticity \textbf{98}(1), 115--116 (2010).
\newblock \doi{10.1007/s10659-009-9230-4}.
\newblock \urlprefix\url{http://dx.doi.org/10.1007/s10659-009-9230-4}

\bibitem{LM}
L{\'e}ger, A., Miara, B.: The obstacle problem for shallow shells: curvilinear
  approach.
\newblock Int. J. Numer. Anal. Model. Ser. B \textbf{2}(1), 1--26 (2011)

\bibitem{Lions}
Lions, J.L.: Perturbations singuli\`eres dans les probl\`emes aux limites et en
  contr\^ole optimal.
\newblock Lecture Notes in Mathematics, Vol. 323. Springer-Verlag, Berlin-New
  York (1973)

\bibitem{MM}
Martins, J., Marques, M.M. (eds.): Contact Mechanics.
\newblock Kluwer Academic Publishers, Dordrecht (2001)

\bibitem{RJM}
Raous, M., Jean, M., Moreau, J. (eds.): Contact Mechanics.
\newblock Plenum Press, New York (1995)

\bibitem{ASV15}
Rodr{\'{\i}}guez-Ar{\'o}s, A., Via{\~n}o, J.M.: A bending-streching model in
  adhesive contact for elastic rods obtained by using asymptotic methods.
\newblock Nonlinear Anal. Real World Appl. \textbf{22}, 632--644 (2015)

\bibitem{Sh}
Shillor, M. (ed.): Recent Advances in Contact Mechanics, vol. 28 (4--8).
\newblock Math.\ Computer Model (1998)

\bibitem{SST}
Shillor, M., Sofonea, M., Telega, J.J.: Models and Analysis of Quasistatic
  Contact, \emph{Lecture Notes in Physics}, vol. 655.
\newblock Springer, Berlin (2004)

\bibitem{TraViano}
Trabucho, L., Via{\~n}o, J.M.: Mathematical modelling of rods.
\newblock In: Handbook of numerical analysis, {V}ol.\ {IV}, Handb. Numer.
  Anal., IV, pp. 487--974. North-Holland, Amsterdam (1996)

\bibitem{Viano1}
Via{\~n}o, J.M.: The one-dimensional obstacle problem as approximation of the
  three-dimensional {S}ignorini problem.
\newblock Bull. Math. Soc. Sci. Math. Roumanie (N.S.) \textbf{48(96)}(2),
  243--258 (2005)

\bibitem{VAS_wear}
Via{\~n}o, J.M., Rodr{\'{\i}}guez-Ar{\'o}s, {\'A}., Sofonea, M.: Asymptotic
  derivation of quasistatic frictional contact models with wear for elastic
  rods.
\newblock J. Math. Anal. Appl. \textbf{401}(2), 641--653 (2013).
\newblock \doi{10.1016/j.jmaa.2012.12.064}.
\newblock \urlprefix\url{http://dx.doi.org/10.1016/j.jmaa.2012.12.064}

\bibitem{Miara16}
Yan, G., Miara, B.: Mathematical justification of the obstacle problem in the
  case of piezoelectric plate.
\newblock Asymptot. Anal. \textbf{96}(3-4), 283--308 (2016).
\newblock \doi{10.3233/ASY-151339}.
\newblock \urlprefix\url{http://dx.doi.org/10.3233/ASY-151339}

\end{thebibliography}

%
%

\end{document}